\documentclass[11pt]{article}
\usepackage[T1]{fontenc}
\usepackage{amsmath,amsxtra,amssymb,latexsym,amscd,amsthm,pb-diagram,fancyhdr,euscript}
\usepackage{amsthm}
\usepackage{indentfirst}
\usepackage{epsfig}
\usepackage{mathrsfs}
\usepackage{slashed}
\usepackage{hyperref}
\usepackage{float}
\hypersetup{
   colorlinks,
    citecolor=blue,
    filecolor=black,
    linkcolor=blue,
    urlcolor=magenta
}
\hypersetup{linktocpage}

\renewcommand{\d}{\mathrm{d}}

\newcommand{\scri}{{\mathscr I}}

\newcommand{\hook}{{\setlength{\unitlength}{11pt}   
                   \begin{picture}(.833,.8)
                   \put(.15,.08){\line(1,0){.35}}
                   \put(.5,.08){\line(0,1){.5}}
                   \end{picture}}}
\newtheorem{definition}{Definition}
\newtheorem{theorem}{Theorem}
\newtheorem{proposition}{Proposition}

\newtheorem{lemma}{Lemma}
\newtheorem{remark}{Remark}

\newtheorem{defn}{Definition}[section]
\newtheorem{thm}{Theorem}[section]

\newtheorem{cor}{Corollary}[section]
\newtheorem{lem}{Lemma}[section]

\begin{document}
\mbox{} \thispagestyle{empty}

\begin{center}
\bf{\Huge Conformal scattering theory for the Dirac field on Kerr spacetime} \\

\vspace{0.1in}

{PHAM Truong Xuan\footnote{Faculty of Information Technology, Department of Mathematics, Thuyloi university, Khoa Cong nghe Thong tin, Bo mon Toan, Dai hoc Thuy loi, 175 Tay Son, Dong Da, Ha Noi, Viet Nam.\\
Email~: xuanpt@tlu.edu.vn or phamtruongxuan.k5@gmail.com}}
\end{center}

{\bf Abstract.} We investigate to construct a conformal scattering theory of the spin-$1/2$ massless Dirac equation on the Kerr spacetime by using the conformal geometric method and under an assumption on the pointwise decay of the Dirac field. In particular, our construction is valid in the exteriors of Schwarzschild and very slowly Kerr black hole spacetimes, where the pointwise decay was established. 



\tableofcontents

\section{Introduction}
The spin-$1/2$ massless Dirac field were studied since 1960's in the works of Sachs \cite{Sa61} and Penrose \cite{Pe1964,Pe1965} on the ''peeling-off property'' in the Minkowski spacetime. Peeling is the asymptotic behaviours of the spin zero rest-mass fields along the outgoing null geodesic lines. Then, the peeling of Dirac fields were extended to study in the asymptotic flat spacetimes such as Schwarzschild and Kerr spacetimes in \cite{MaNi2012,Xuan2020}. 

The pointwise decays (also called Price's law) of the Dirac and the generalized spin-$n/2$ zero rest-mass fields in the Minkowski spacetime were established by Andersson et al. in \cite{ABJ} by analyzing Hertz potentials. 
The pointwise decays for the Dirac field on the Schwarzschild spacetime were studied by Smoller and Xie \cite{SmolXi} and these results are improved in the recent work of Ma \cite{Ma2021}. Another interesting aspect of the massless spin fields is the local integral formula which were establish initially in the Minkowski spacetime by Penrose \cite{Pe63}. Then, Joudioux \cite{Jo2011} extended the formula on the general curved spacetimes. The local integral formula gives the solution of the Goursat problem in the region near the timlike infinity $i^\pm$ of the Minkowski spacetime.

Concerning the scattering for the massless and massive Dirac equations Nicolas \cite{Ni95}, Nicolas and H\"afner \cite{HaNi2004}, Batic \cite{Ba} and Nicolas et al. in recent work \cite{HNM} established the fully results in the Schwarzschild, Kerr spacetimes and in the interior of the Reissner-Nordstr\"om-type black hole respectively by using the Cook's method and the Mourre estimates, i.e, the completely analytic scattering. In a recent work \cite{Mo2021} Mokdad has obtained the conformal Scattering, i.e, the geometric scattering for Dirac equation in the interior of charged spherically symmetric black holes. 

A conformal scattering theory on the asymptotic flat spacetimes consists on three steps: first, the resolution of the Cauchy problem of the rescaled equations on the rescaled spacetime, then the definition and the extension of the trace operator. Second, we prove the energy identity up to the future timelike infinity $i^+$, this shows that the extension of the trace operator is injective. Third, we solve the Goursat problem with the initial data on the conformal boundary consisting the horizon $\mathfrak{H}^+$ and null infinity $\scri^+$. Therefore, the trace operator will be surjective, hence, an isometry.  

Concerning the Goursat problem of the massless spin equations, Mason and Nicolas established the well-posed of the scalar wave, Dirac and Maxwell equations in the asymptotic simple spacetimes in \cite{MaNi2004}. By using these results, they constructed the conformal scattering operators, i.e, the geometric scattering operators for these field equations in the asymptotic simple spacetimes. After that, the Goursat problem for the spin field equations in the asymptotic flat spacetime are also established in some recent works. In particular, the Goursat problem for the scalar wave equation has solved in the Schwarzschild by Nicolas \cite{Ni2016} and the ones for Dirac and Maxwell equations in the Reissner-Nordstr\"om-de Sitter spacetime have treated by Mokdad \cite{Mo2019,Mo2021}. The Goursat prolem for the generalized $n/2$-spin zero ress-mass equations have studied by the author in \cite{Xuan21}.
The method combines the vector field method (use the pointwise decays to establish the energy estimates and energy equality) and the generalized results of H\"ormander for the scalar and spin wave equations (see \cite{Mo2019,Mo2021,Ni2016,Xuan21}). 
The Goursat problem is an important step to construct the conformal scattering theory, i.e, the geometric scattering theory for the field equations in the asymptotic simple and flat spacetimes (in detail see \cite{MaNi2004,Ni2016} for the initial construction of the theory).

In the present paper, we will extend \cite{Xuan21} to establish the conformal scattering theory for the spin-$1/2$ massless Dirac equation in the Kerr spacetime by using conformal geometric methods. We use the Penrose's conformal mapping on the exterior of Kerr to obtain the domain $\bar{\mathcal{B}}_I$ which consits three singular points: the spacelike infinity $i_0$ and the future (past) timelike infinity $i^\pm$. We will construct the conformal scattering for the Dirac equation on this domain. 
Since the Cauchy problem for the Dirac equation is well-posedness on the exterior of Kerr black holde (see \cite{Ni2002}) we can define the trace operator of the Dirac field on the conformal boundaries $\mathfrak{H}^+\cup\scri^+$.
We will close the timelike future infinity $i^+$ by a hyperboloid hypersurface $\mathcal{H}_T$.
Under an assumption on the decay rate (follows both the time and spatial directions) of the components of Dirac field we prove that the energy of the Dirac field through $\mathcal{H}_T$ tends to zero. This leads to the energy equality of the Dirac field between the initial hypersurface $\Sigma_0$ and the null conformal boundaries $\mathfrak{H}^+\cup \scri^+$. As a consequent we can extend the trace operator on the Sobolev space on $\mathfrak{H}^+\cup \scri^+$ with the energy norm by density.

We develop the methods in \cite{Mo2019,Ni2016,Xuan2021,Xuan20,Xuan21} to establish the well-posed of the Goursat problem. In particular, the Goursat problem will be solved by using the energy equality and the generalized results of H\"ormander in two parts. The first one we apply the generalized results of H\"ormander to obtained the solution in the future $\mathcal{I}^+(\mathcal{S})$ of the Cauchy hypersurface $\mathcal{S}$ which pass the bifurcation sphere and intersects strictly at the past of the support of initial data. The second one we extend the solution of the first part down to the initial hypersurface $\Sigma_0$ by using again the well-posed of the Cauchy problem and the equality energy. The solution of the Goursat problem is an union of the ones obtained in the two parts $\mathcal{I}^+(\mathcal{S})$ and $\mathcal{I}^-(\mathcal{S})$ (see Theorem \ref{Goursatprob}).

The paper is organized as follows: Section \ref{S2} we recall the geometric setting of the Kerr spacetime and some formulas of the curvature spinors and their commutators; Section \ref{S4} we study the spin-$1/2$ massless Dirac equation and prove the energy equality on the rescaled Kerr spacetime; Section \ref{S5} we solve the Goursat problem and establish the conformal scattering operator which associates the past scattering data to the future scattering data; Appendix \ref{A} we provide the applying of the generalized results of H\"ormander's for the spin wave equation on Kerr spacetime, then establish a pointwise decay of the components of Dirac field on very slowly Kerr spacetimes and finally Appendix \ref{B} we give some calculations which are necessary to solve the Goursat problem, the detailed components of Riemann curvature and the covariant derivetives of origin spin-frame.\\
{\bf Remarks and Notations.}\\
$\bullet$ We use the formalisms of abstract indices, $2$-component spinors and Newman-Penrose and Geroch-Held-Penrose.\\
$\bullet$ Let $f(x)$ and $g(x)$ be two real functions. We write $f \lesssim g$ if there exists a constant $C \in (0,+\infty)$ such that $f(x)\leq C g(x)$ for all $x$, and write $f\simeq g$ if both $f\lesssim g$ and $g\lesssim f$ are valid.\\
$\bullet$ This work is completely a part of the conformal scattering for Dirac field on Kerr metric focus on the proof of the energy equality and the well-posedness of the Goursat problem. There is other important part which relates to the pointwise decay of the field on the full subextremal Kerr spactime. In this work we take the pointwise decay into the decay assumption \eqref{pointwise}. The decay assumption \eqref{pointwise} are valid for the case $a=0$, Schwarzschild spacetime (see \cite[Theorem 1.6]{Ma2021}) and very slowly Kerr spacetimes (see Appendix \ref{verySlowKer}).\\
{\bf Acknowledgements.} 
The author is grateful to Siyuan Ma for the explanation on the Price's law of the Dirac field on Kerr spacetimes.

\section{Geometric and analytic setting} \label{S2}
\subsection{Kerr metric, star-Kerr and Kerr-star coordinates}
In Boyer-Lindquist coordinates, Kerr spacetime is a manifold $({\cal M}={\mathbb R}_t\times {\mathbb R}_r \times S_\omega^2, \, g)$ whose the metric $g$ takes the form
\begin{equation}\label{Metric}
 g=\left(1-\frac{2Mr}{\rho^2}\right)\d t^2 +\frac{4aMr\sin^2\theta}{\rho^2}\d t\d \varphi -\frac{\rho^2}{\Delta}\d r^2-\rho^2\d \theta^2-\frac{\sigma^2}{\rho^2}\sin^2\theta\d \varphi^2, 
\end{equation}
$$ \rho^2=r^2+a^2\cos^2\theta,\Delta=r^2-2Mr+a^2, $$
$$ \sigma^2=(r^2+a^2)\rho^2+2Mra^2\sin^2\theta=(r^2+a^2)^2-\Delta a^2\sin^2\theta,$$
where $M>0$ is the mass of the black hole and $a\neq 0$ is its angular momentum per unit mass. Kerr spacetime is asymptotically plat and there are only two basic Killing vectors $\partial_t$ and $\partial_{\varphi}$. In this paper, we work on the exterior of the black hole $\mathcal{B}_I=\left\{ r>r_+ = M + \sqrt{M^2-a^2} \right\}$.

The non zero components of the Riemann curvature tensor of the Kerr metric are given in Appendix \ref{Riemann}.

The determinant of $g$ is given by $\det{g} = -\rho^4\sin^2\theta$, and we give here the form of the inverse Kerr metric $g^{-1}$, which will be useful to the next sections
\begin{equation}\label{inverse}
g^{-1} = \frac{1}{\rho^2} \left( \frac{\sigma^2}{\Delta}\partial_t^2 - \frac{2aMr}{\Delta}\partial_t\partial_{\varphi} - \Delta\partial_r^2 - \partial_{\theta}^2 - \frac{\rho^2 - 2Mr}{\Delta\sin^2\theta}\partial_{\varphi}  \right).
\end{equation}

The Kerr spacetime has two principal null geodesics $V^{\pm}$:
\begin{equation*}
V^\pm = \frac{r^2+a^2}{\Delta}\partial_t \pm \partial_r + \frac{a}{\Delta}\partial_\varphi.
\end{equation*}
These geodesics lead to construct the star-Kerr and Kerr-star coordinates. In particular, the star-Kerr coordinates $({^*t},r,\theta,{^*\varphi})$ are defined as
$${^*t} = t - r_*, \, {^*\varphi} = \varphi - \Lambda(r),$$
where the function $r_*$ is Regge-Wheeler-type variable and $\Lambda$ sastifies
$$\frac{\d r_*}{\d r} = \frac{r^2+a^2}{\Delta}, \,  \frac{\d \Lambda}{\d r} = \frac{a}{\Delta} \, .$$
In these coordinates, the outgoing principal null geodesics are the $r$-coordinate lines and the metric has the form
$$g = \left( 1 - \frac{2Mr}{\rho^2} \right)\d {^*t}^2 + \frac{4aMr\sin^2\theta}{\rho^2}\d{^*t}\d{^*\varphi} - \rho^2\d\theta^2 - \frac{\sigma^2}{\rho^2}\sin^2\theta\d{^*\varphi} + 2\d{^*t}\d r - 2a\sin^2\theta\d{^*\varphi}\d r.$$
In this coordinate system we can add the horizon
$${\mathfrak H}^- = {\mathbb R}_{^*t}\times \left\{r=r_+\right\}\times S_{\theta,{^*\varphi}}^2$$
as a smooth null boundary to $\mathcal{B}_I$. 
\begin{remark}
We have the following relations between the Boyer-Lindquist ($BL$-) coordinates and star-Kerr ($*K$-) coordinates
\begin{eqnarray*}
\partial_t &=& \partial_{^*t}\cr
\partial_\theta &=& (\partial_\theta)_{*K}\cr
\partial_\varphi &=& \partial_{^*\varphi}\cr
(\partial_r)_{BL} &=& -\frac{r^2+a^2}{\Delta}\partial_{^*t} + (\partial_r)_{*K} - \frac{a}{\Delta}\partial_{^*\varphi}. 
\end{eqnarray*}
\end{remark}
Similarly star-Kerr coordinates, Kerr-star coordinates $({^*t},r,\theta,{^*\varphi})$ are constructed by using the incoming principal null geodesics that are parametrized as the integral lines of $V^-$. We have
$${t^*} = t + r_*, \, {\varphi^*} = \varphi + \Lambda(r) \; ,$$
with the same function $r_*$ and $\Lambda$ as for star-Kerr coordinates. Consequently the incoming principal null geodesics can be considered as the $r$ coordinates curves parametrized by $s = r$.

Kerr-star coordinates are defined globally on block I. In Kerr-star coordinate system the Kerr metric takes the form
\begin{equation*}
g =\left(1-\frac{2Mr}{\rho^2}\right) \d {t^*}^2 + \frac{4aMr\sin^2\theta}{\rho^2} \d {t^*}\d {\varphi^*} -\frac{\sigma^2}{\rho^2}\sin^2\theta \d {\varphi^*}^2 - \rho^2 \d \theta^2 - 2 \d {t^*} \d r + 2 a \sin^2\theta \d {\varphi^*} \d r.
\end{equation*}
This coordinate system allows us to add the horizon
$${\mathfrak H}^+ = {\mathbb R}_{t^*}\times \left\{r=r_+\right\}\times S_{\theta,{\varphi^*}}^2$$
as a smooth null boundary to $\mathcal{B}_I$. 
\begin{remark}
We have the following relations between the Boyer-Lindquist ($BL$-) coordinates and Kerr-star ($K*$-) coordinates
\begin{eqnarray*}
\partial_t &=& \partial_{t^*}\cr
\partial_\theta &=& (\partial_\theta)_{K*}\cr
\partial_\varphi &=& \partial_{\varphi^*}\cr
(\partial_r)_{BL} &=& \frac{r^2+a^2}{\Delta}\partial_{t^*} + (\partial_r)_{K*} + \frac{a}{\Delta}\partial_{\varphi^*}. 
\end{eqnarray*}
\end{remark}

\subsection{Penrose's conformal compactification}
The Penrose compactification of Kerr spacetime is constructed by using the star-Kerr and Kerr-star coordinates. 
Now we consider the rescaled star-Kerr coordinates $({^*t}, \, R =1/r,\, \theta,\, {^*\varphi})$, and rescale the Kerr metric with the conformal factor $\Omega=R$, as follows
\begin{eqnarray}
\widetilde{g}: = R^2 g&=& R^2\left(1-\frac{2Mr}{\rho^2}\right)\d {^*t}^2+\frac{4MaR\sin^2\theta}{\rho^2}\d {^*t} \d {^*\varphi} \cr
&& -(1+a^2R^2\cos^2\theta)\d {\theta^2} -\left(1+a^2R^2+\frac{2Ma^2R\sin^2\theta}{\rho^2}\right)\sin^2\theta \d {^*\varphi^2}\cr
&&-2 \d {^*t} \d R+2a\sin^2\theta \d {^*\varphi} \d R. \label{ResMetric}
\end{eqnarray}
The rescaled metric extends smoothly and non degenerately to the null hypersurface $\scri^+ = \mathbb{R}_{^*t}\times \left\{ R = 0 \right\} \times S_{\omega}^2$, which can be added as a smooth boundary to Kerr spacetime and will be called future null infinity. The Levi-Civita symbols must be rescaled as 
$$\widetilde{\varepsilon}_{AB}=R\varepsilon_{AB}.$$
The rescaled metric has the inverse form
\begin{eqnarray}\label{InvertResMetric}
{\widetilde g}^{-1}&=&-\frac{1}{\rho^2}\left(r^2a^2\sin^2\theta\partial^2_{^*t}+2(r^2+a^2)\partial_{^*t}\partial_R+2ar^2\partial_{^*t}\partial_{^*\varphi}+2a\partial_R\partial_{^*\varphi} \right) \cr
&&-\frac{1}{\rho^2} \left(R^2\Delta\partial^2_R+r^2\partial_{\theta}^2+\frac{r^2}{\sin^2\theta}\partial^2_{^*\varphi}\right). 
\end{eqnarray}

We now consider the rescaled Kerr-star coordinates $(t^*, \, R =1/r,\, \theta,\, \varphi^*)$, and rescale the Kerr metric with the conformal factor $\Omega^2=R^2$, as follows
\begin{eqnarray}
\widehat{g}: = R^2 g&=& R^2\left(1-\frac{2Mr}{\rho^2}\right)\d {t^*}^2+\frac{4MaR\sin^2\theta}{\rho^2}\d {t^*} \d {\varphi^*} \cr
&& -(1+a^2R^2\cos^2\theta)\d {\theta^2} -\left(1+a^2R^2+\frac{2Ma^2R\sin^2\theta}{\rho^2}\right)\sin^2\theta \d {{\varphi^*}^2}\cr
&&+2 \d {t^*} \d R+2a\sin^2\theta \d {\varphi^*} \d R. \label{ResMetric}
\end{eqnarray}
The rescaled metric extends smoothly and non degenerately to the null hypersurface $\scri^- = \mathbb{R}_{t^*}\times \left\{ R = 0 \right\} \times S_{\omega}^2$, which can be added as a smooth boundary to Kerr spacetime and will be called future null infinity. 

Therefore, the Penrose conformal compactification of $\mathcal{B}_I$ is
$$\bar{\mathcal{B}_I} = \mathcal{B}_I \cup \scri^+ \cup \mathfrak{H}^+ \cup \scri^-\cup \mathfrak{H}^-\cup S_c^2,$$
where $S_c^2$ is the bifurcation sphere.
\begin{remark}
We notice that the compactified space-time is not closed, it consists three singularities ''points'' of the rescaled metric $\widetilde{g}$ (and $\widehat{g}$): future timlike infinity $i^+$ is defined as the limit point of uniformly timelike curves as $t\to + \infty$, past timelike infinity $i^-$ is symmetric of $i^+$ in the distant past and spacelike infinity $i_0$ is the limit point of uniformly spacelike curves as $r\to +\infty$. 
\end{remark}
We will use the normalized Newman-Penrose tetrad given by H\"afner and Nicolas in \cite{HaNi2004}. More precisely
\begin{equation}\label{New2}
\begin{gathered}
l^a\partial_a = \frac{1}{\sqrt{2\Delta\rho^2}} \left( (r^2+a^2)\partial_t + \Delta \partial_r + a \partial_\varphi \right), \cr
n^a \partial_a = \frac{1}{\sqrt{2\Delta\rho^2}} \left( (r^2+a^2)\partial_t - \Delta \partial_r + a \partial_\varphi \right), \cr
m^a\partial_a = \frac{1}{p\sqrt 2}\left(  ia\sin\theta \partial_t + \partial_\theta + \frac{i}{\sin\theta} \partial_\varphi  \right), \, \mbox{where} \,\,  p = r + ia \cos\theta.
\end{gathered}
\end{equation}
The dual tetrad of a $1-$form is given as follows
\begin{equation}\label{DualNew2}
\begin{gathered}
l_a \d x^a = \sqrt{\frac{\Delta}{2\rho^2}} \left( \d t - \frac{\rho^2}{\Delta} \d r - a\sin^2\theta \d \varphi \right),\cr
n_a \d x^a = \sqrt{\frac{\Delta}{2\rho^2}} \left( \d t + \frac{\rho^2}{\Delta} \d r - a\sin^2\theta \d \varphi \right),\cr
m_a \d x^a = \frac{1}{p\sqrt 2} \left( ia\sin\theta \d t - \rho^2 \d \theta - i(r^2+a^2)\sin\theta \d \varphi \right).
\end{gathered}
\end{equation}

We define the rescaled tetrad $\widetilde{l}^a, \, \widetilde{n}^a, \, \widetilde{m}^a,\, \bar{\widetilde{m}}^a,$ which is normalized with respect to $\widetilde{g}$, as follows:
$$ \widetilde l^a = r^2 l^a, \, \widetilde n^a = n^a, \, \widetilde m^a = rm^a $$
and we have
$$ \widetilde{l}_a = l_a, \, \widetilde{n}_a = R^2n_a, \, \widetilde{m}_a = Rm_a.$$
Since
$$(\partial_r)_{BL} = -\frac{r^2+a^2}{\Delta}\partial_{^*t} + (\partial_r)_{*K} - \frac{a}{\Delta}\partial_{^*\varphi},$$
we have
$$r^2\Delta(\partial_r)_{BL} = -r^2(r^2+a^2)\partial_{^*t} - \Delta\partial_R - ar^2\partial_{^*\varphi}.$$
Therefore, in the rescaled star-Kerr coordinates $({^*t}, \, R, \, \theta, \, {^*\varphi})$, the rescaled Newman Penrose tetrad becomes
\begin{equation}\label{resNew2}
\begin{gathered}
\widetilde{l}^a \partial_a = -\sqrt\frac{\Delta}{2\rho^2}\partial_R,\cr
\widetilde{n}^a \partial_a = \sqrt\frac{2}{\Delta \rho^2} \left( (r^2+a^2)\partial_{^*t}+a\partial_{^*\varphi}+\frac{R^2\Delta}{2}\partial_R \right),\cr
\widetilde{m}^a \partial_a = \frac{r}{p\sqrt{2}} \left(ia\sin\theta\partial_{^*t}+ \partial_\theta + \frac{i}{\sin \theta} \partial_{^*\varphi} \right) 
\end{gathered} 
\end{equation}
and
\begin{equation}\label{dualNew3}
\begin{gathered}
\widetilde{l}_a \d x^a = \sqrt{\frac{\Delta}{2\rho^2}} \left( \d {^*t} - a\sin^2\theta \d {^*\varphi} \right),\cr
\widetilde{n}_a \d x^a = \sqrt{\frac{\Delta}{2\rho^2}} \left( R^2 \d {^*t} - \frac{2\rho^2}{\Delta}\d R - aR^2\sin^2\theta \d {^*\varphi} \right),\cr
\widetilde{m}_a \d x^a = \frac{1}{p\sqrt 2} \left( iaR\sin\theta \d {^*t} - R\rho^2 \d \theta - iR(r^2+a^2)\sin\theta \d {^*\varphi} \right).
\end{gathered}
\end{equation}
In terms of the associated spin-frame $\left\{ o^A,\iota^A \right\}$, the above relation is equivalent to the following rescaling
$$ \widetilde{o}^A = r o^A, \, \widetilde{\iota}^A = \iota^A, \, \widetilde{o}_A = o_A, \, \widetilde{\iota}_A = R\iota_A.$$

Simimarly, in the rescaled Kerr-star coordinates $({t^*}, \, R, \, \theta, \, {\varphi^*})$, the rescaled Newman Penrose tetrad becomes
\begin{equation}\label{resNew22}
\begin{gathered}
\widehat{l}^a \partial_a =  \sqrt\frac{2}{\Delta \rho^2} \left( (r^2+a^2)\partial_{t^*}+a\partial_{\varphi^*}-\frac{R^2\Delta}{2}\partial_R \right),\cr
\widehat{n}^a \partial_a = \sqrt\frac{\Delta}{2\rho^2}\partial_R,\cr
\widehat{m}^a \partial_a = \frac{r}{p\sqrt{2}} \left(ia\sin\theta\partial_{t^*}+ \partial_\theta + \frac{i}{\sin \theta} \partial_{\varphi^*} \right) 
\end{gathered} 
\end{equation}
and
\begin{equation}\label{dualNew33}
\begin{gathered}
\widehat{l}_a \d x^a = \sqrt{\frac{\Delta}{2\rho^2}} \left( R^2 \d {t^*} + \frac{2\rho^2}{\Delta}\d R - aR^2\sin^2\theta \d {\varphi^*} \right) ,\cr
\widehat{n}_a \d x^a = -\sqrt{\frac{\Delta}{2\rho^2}} \left( \d {t^*} - a\sin^2\theta \d {\varphi^*} \right),\cr
\widehat{m}_a \d x^a = \frac{1}{p\sqrt 2} \left( iaR\sin\theta \d {t^*} - R\rho^2 \d \theta - iR(r^2+a^2)\sin\theta \d {\varphi^*} \right).
\end{gathered}
\end{equation}
We have also the following relations
$$ \widehat{l}^a = l^a, \, \widehat{n}^a = r^2n^a, \, \widehat{m}^a = rm^a, $$
$$ \widehat{l}_a = R^2 l_a, \, \widehat{n}_a = n_a, \, \widehat{m}_a = Rm_a,$$
$$ \widehat{o}^A = o^A, \, \widehat{\iota}^A = r\iota^A, \, \widehat{o}_A = R o_A, \, \widehat{\iota}_A = \iota_A.$$

The $4$-volume measure associated with the rescaled Kerr metric $\widetilde{g}$ is
\begin{eqnarray*}
\mathrm{dVol}_{\widetilde{g}} &=& i \widetilde{l} \wedge \widetilde{n} \wedge \widetilde{m} \wedge \bar{\widetilde{m}} = R^2\rho^2 \d {^*t} \wedge \d R \wedge \d^2 \omega\cr
&=& \frac{R^2\Delta}{a^2+r^2} \d t \wedge \d r_* \wedge \d^2\omega = \frac{R^2\Delta}{2(a^2+r^2)}\d {^*t} \wedge \d {t^*} \wedge \d^2\omega,
\end{eqnarray*}
where $\d^2\omega$ is the euclidean area element on unit $2$-sphere $\mathbb{S}^2$.

\subsection{Curvature spinors}\label{app_1}
On Kerr spacetime $({\cal M},g)$, we recall that the Riemann tensor $R_{abcd}$ can be decomposed as follows (see Equation (4.6.1) page 231 in R. Penrose and W. Rindler \cite[Vol. 1]{PeRi}): 
\begin{equation}\label{decomposition_Riemann}
 R_{abcd} = X_{ABCD} \, \varepsilon_{A'B'} \varepsilon_{C'D'} + \Phi_{ABC'D'} \, \varepsilon_{A'B'} \varepsilon_{CD} + \bar{\Phi}_{A'B'CD} \, \varepsilon_{AB} \varepsilon_{C'D'} + \bar{X}_{A'B'C'D'} \, \varepsilon_{AB} \varepsilon_{CD},
\end{equation} 
where $X_{ABCD}$ is a complete contraction of the Riemann tensor in its primed spinor indices
$$ X_{ABCD} = \frac{1}{4} R_{abcd} {\varepsilon}^{A'B'} {\varepsilon}^{C'D'},$$
and ${\Phi}_{ab} = {\Phi}_{(ab)}$ is the trace-free part of the Ricci tensor multiplied by $-1/2$~:
$$2{\Phi}_{ab} = 6 {\Lambda} {g}_{ab} - {R}_{ab} \, ,~ {\Lambda} = \frac{1}{24} \mathrm{Scal}_{g}.$$

We set
$$P_{ab} = \Phi_{ab} - \Lambda g_{ab},$$
$$ X_{ABCD} = \Psi_{ABCD} + \Lambda \left( \varepsilon_{AC} \varepsilon_{BD} + \varepsilon_{AD} \varepsilon_{BC} \right),\, \Psi_{ABCD} = X_{(ABCD)} = X_{A(BCD)}.$$

Under a conformal rescaling $\widetilde{g}=\Omega^2 g$ we have (see R. Penrose and W. Rindler \cite[Vol. 2]{PeRi})
\begin{gather*}
\widetilde{\Psi}_{ABCD} = \Psi_{ABCD}, \cr
\widetilde{\Lambda} = \Omega^{-2} \Lambda + \frac{1}{4} \Omega^{-3} \square \Omega, ~\square = \nabla^a \nabla_a, \cr
\widetilde{P}_{ab} = P_{ab} - \nabla_b \Upsilon_a  + \Upsilon_{AB'} \Upsilon_{BA'},~ \mbox{with } \Upsilon_a = \Omega^{-1} \nabla_a \Omega = \nabla_a \log \Omega.
\end{gather*}

In the rescaled Kerr spacetime we have the detailed calculations of $\widetilde{\Psi}_{ABCD}$ and $\widetilde{X}_{ABCD}$ in the following proposition.
\begin{lemma}\label{CurSpin}
For the Kerr metric \eqref{Metric} we have $\Lambda = \Phi_{ab} = 0$ and the components of $X_{ABCD}=\Psi_{ABCD}$ are regular. For the rescaled Kerr metric \eqref{ResMetric} we have
$$\widetilde\Lambda=\frac{Mr-a^2}{2\rho^2}$$
then the components of $\widetilde{X}_{ABCD}$ are regular and the simpler expression of the rescaled trace-free Ricci tensor is
\begin{eqnarray*}\label{Cur}
\widetilde{\Phi}_{ab}\d x^a\d x^b &=& A_1 \d{^*t}^2 + A_2 \d{^*t}\d R + A_3 \d{^*t}\d\theta + A_4\sin^2\theta\d{^*t}\d{^*\varphi} + A_5 \d R^2 + A_6\d R \d \theta \cr
&& + A_7\sin^2\theta \d R \d {^*\varphi} + A_8\d \theta^2 + A_9\sin^2\theta\d \theta\d {^*\varphi} + A_{10}\sin^2\theta\d{^*\varphi}^2, 
\end{eqnarray*}
where $A_i(r,\theta), \, i=1,2...10$ are regular and bounded functions.
\end{lemma}
\begin{proof}
Since the Kerr is Ricci flat, we have $\Lambda = \Phi_{ab} = 0$. The equaltites \eqref{Cur} were proven in \cite[Lemma 4.1]{Xuan2020}. Finally, the regularity in $r$ and $\theta$ of the components of $X_{ABCD}=\Psi_{ABCD}$ can be verified directly from the non zero components of the Riemann curvature tensor (see Appendix \ref{Riemann}) and the inverse Kerr metric \eqref{inverse}. As a consequence the components of
$$\widetilde{X}_{ABCD} = X_{ABCD} + \widetilde{\Lambda} \left( \widetilde{\varepsilon}_{AC} \widetilde{\varepsilon}_{BD} + \widetilde{\varepsilon}_{AD} \widetilde{\varepsilon}_{BC} \right)$$
are regular.
\end{proof}

\subsection{Spinor form of commutators}\label{app_3_Commutator}
We recall some basic formulas on the spin wave operator which act on the spin-$1/2$ fields on $(\bar{\mathcal{B}}_I,\widetilde{g})$ (for the generalized spin fields see \cite[Appendix]{Xuan21}). Since the anti-symmetric property of $\widetilde{\Delta}^{ab} = \widetilde{\nabla}^{[a}\widetilde{\nabla}^{b]}$, we have
$$\widetilde{\Delta}^{ab} = 2\widetilde{\nabla}^{[a}\widetilde{\nabla}^{b]} = \widetilde{\varepsilon}^{A'B'} \widetilde{\Box}^{AB} + \widetilde{\varepsilon}^{AB} \widetilde{\Box}^{A'B'},$$
where
$$\widetilde{\Box}^{AB} = \widetilde{\nabla}^{X'(A}\widetilde{\nabla}^{B)}{_{X'}}, \; \widetilde{\Box}^{A'B'} = \widetilde{\nabla}^{X(A'}\widetilde{\nabla}^{B')}{_{X}}.$$
Now we have
$$\widetilde{\Delta}^{ab} = \widetilde{g}^{ac}\widetilde{g}^{bd} \widetilde{\Delta}_{cd},$$
and $\widetilde{\Delta}_{ab}$ acts on the spinor form $\kappa^C $ as
$$\widetilde{\Delta}_{ab}\kappa^C = \left\{ \widetilde{\varepsilon}_{A'B'} \widetilde{X}_{ABE}{^C} + \widetilde{\varepsilon}_{AB} \widetilde{\Phi}_{A'B'E}{^{C}} \right\} \kappa^E,$$
where $\widetilde{X}_{ABCD}$ and $\widetilde{\Phi}_{ABC'D'}$ are the curvature spinors in the expression of the Riemann tensor $\widetilde{R}_{abcd}$~:
$$\widetilde{R}_{abcd} = \widetilde{X}_{ABCD} \, \widetilde{\varepsilon}_{A'B'} \widetilde{\varepsilon}_{C'D'} + \widetilde{\Phi}_{ABC'D'} \, \widetilde{\varepsilon}_{A'B'} \widetilde{\varepsilon}_{CD} + \bar{\widetilde{\Phi}}_{A'B'CD} \, \widetilde{\varepsilon}_{AB} \widetilde{\varepsilon}_{C'D'} + \bar{\widetilde{X}}_{A'B'C'D'} \, \widetilde{\varepsilon}_{AB} \widetilde{\varepsilon}_{CD}.$$
Hence, we obtain
$$\widetilde{\Delta}^{ab} \kappa^C = \widetilde{\varepsilon}^{AC} \widetilde{\varepsilon}^{A'C'} \widetilde{\varepsilon}^{BD} \widetilde{\varepsilon}^{B'D'} \widetilde{\Delta}_{cd}\kappa^C = \left\{ \widetilde{\varepsilon}^{A'B'} \widetilde{X}^{AB}{_E}{^C} + \widetilde{\varepsilon}^{AB} \widetilde{\Phi}^{A'B'}{_E}{^{C}} \right\} \kappa^E.$$
By taking symmetrizing and skew-symmetrizing over $AB$ we get the equations
$$\widetilde{\Box}^{AB} \kappa^C =  \widetilde{X}^{AB}{_E}{^C}  \kappa^E, \; \widetilde{\Box}^{A'B'}\kappa^{C} =  \widetilde{\Phi}^{A'B'}{_E}{^{C}}\kappa^E. $$
Similarly, we can obtain the formula of the primed spin-vectors
$$\widetilde{\Delta}^{ab} \tau^{C '}= \left\{ \widetilde{\varepsilon}^{AB} \bar{\widetilde{X}}^{A'B'}{_{E'}}{^{C'}} + \widetilde{\varepsilon}^{A'B'} \widetilde{\Phi}^{AB}{_{E'}}{^{C'}} \right\} \tau{^{E'}},$$
$$\widetilde{\Box}^{AB} \tau^{C'} =  \widetilde{\Phi}^{AB}{_{E'}}{^{C'}} \tau^{E'}, \; \widetilde{\Box}^{A'B'}\tau^{C'} =  \bar {\widetilde{X}}^{A'B'}{_{E'}}{^{C'}}\tau^{E'}. $$
Lowering the index $C$ (or $C'$), we get
$$\widetilde{\Box}^{AB} \kappa_C =   \widetilde{X}^{ABE}{_C}  \kappa_E, \; \widetilde{\Box}^{A'B'}\kappa_{C} =  \widetilde{\Phi}^{A'B'E}{_{C}}\kappa^E, $$
$$\widetilde{\Box}^{AB} \tau_{C'} =  \widetilde{\Phi}^{AB{E'}}{_{C'}}  \tau_{E'}, \; \widetilde{\Box}^{A'B'}\tau_{C'} =  \bar {\widetilde{X}}^{A'B'{E'}}{_{C'}}\tau_{E'}.$$

Using the above formulas we establish some basic formulas on $(\bar{\mathcal{B}}_I,\widetilde{g})$ which will be useful for the next sections. First, we have
\begin{eqnarray}\label{app_3_wave0}
\widetilde{\nabla}_{ZA'}\widetilde{\nabla}^{AA'} \widetilde{\phi}_A &=& \widetilde{\varepsilon}^{AM}\widetilde{\nabla}_{ZA'}\widetilde{\nabla}_M^{A'} \widetilde{\phi}_A = \widetilde{\varepsilon}^{AM} \left( \widetilde{\nabla}_{A'[Z}\widetilde{\nabla}{_{M]}}{^{A'}} + \widetilde{\nabla}_{A'(Z}\widetilde{\nabla}{_{M)}}{^{A'}}  \right) \widetilde{\phi}_A \cr
&=& \widetilde{\varepsilon}^{AM} \left( \frac{1}{2}\widetilde{\varepsilon}_{ZM} \widetilde{\Box} + \widetilde{\Box}_{ZM}  \right) \widetilde{\phi}_A = \frac{1}{2} \widetilde{\varepsilon}_Z{^A}\widetilde{\Box} \widetilde{\phi}_A + \widetilde{\Box}_{Z}{^A} \widetilde{\phi}_A \cr
&=& \frac{1}{2}\check{\Box} \widetilde{\phi}_Z + \widetilde{X}_{ZA}{^{NA}} \widetilde{\phi}_N. 
\end{eqnarray}
We have also
\begin{eqnarray}\label{app_3_wave1}
\widetilde{\nabla}^{AA'}\widetilde{\nabla}_{AK'}\Xi^{K'} &=& - \widetilde{\varepsilon}_{K'M'}\widetilde{\nabla}^{AA'}\widetilde{\nabla}^{M'}_{A}\Xi^{K'} = - \widetilde{\varepsilon}_{K'M'} \left( \widetilde{\nabla}^{A[A'}\widetilde{\nabla}^{M']}_{A} + \widetilde{\nabla}^{A(A'}\widetilde{\nabla}^{M')}_{A} \right)\Xi^{K'}\cr
&=& - \widetilde{\varepsilon}_{K'M'} \left(  \frac{1}{2} \widetilde{\varepsilon}^{A'M'} \widetilde{\Box} + \widetilde{\Box}^{A'M'} \right)\Xi^{K'} = \frac{1}{2}\widetilde{\varepsilon}^{A'} {_{K'}}\widetilde{\Box}\Xi^{K'} + \widetilde{\Box}^{A'}{_{K'}}\Xi^{K'}\cr
&=& \frac{1}{2}\breve{\Box}\Xi^{A'} + \bar{{\widetilde X}}^{A'}{_{K'Q'}}{^{K'}} \Xi^{Q'}.
\end{eqnarray}

If we define the spin wave operator which acts on the full spin fields by using the spinor form as following
\begin{equation}\label{WaveOperator1}
\widehat{\Box} = \widetilde{\varepsilon}^{MN} \widetilde{\varepsilon}_{M'N'}\widetilde{\nabla}_M^{M'}\widetilde{\nabla}_N^{N'} = \widetilde{\nabla}_a\widetilde{\nabla}^a,
\end{equation}
then we have
\begin{eqnarray}\label{WaveOperator2}
\widehat{\Box} &=& \widetilde{\varepsilon}^{MN}\nabla_{N'M}\nabla_N^{N'} = \widetilde{\varepsilon}^{MN} \left( \nabla_{N'[M}\nabla_{N]}^{N'} + \nabla_{N'(M}\nabla_{N)}^{N'} \right) \cr
&=& \varepsilon^{MN} \left( \frac{1}{2}\varepsilon_{MN}\check{\Box} + \check{\Box}_{MN} \right) \cr
&=& \check{\Box} - \check{\Box}_{M}^M.
\end{eqnarray}
Similarly
\begin{equation}\label{WaveOperator3}
\widehat{\Box} =  \widetilde{\varepsilon}_{M'N'} \left( - \frac{1}{2}\varepsilon^{M'N'}{\breve{\Box}} + {\breve{\Box}}^{M'N'} \right)  = -{\breve{\Box}} + {\breve{\Box}}^{M'}_{M'}.
\end{equation}
\begin{remark}
The spin operators $\check{\Box}$, $\breve{\Box}$ and $\widehat{\Box}$ (they act on the full spin fields $\phi_{A}$)  and the scalar wave operator $\Box_{g}$ defined by Laplcace-Beltrami formula (it acts on the scalar fields $\Phi = (\phi_0,\phi_1)$), are of the same modulo the derivation terms of order less than or equal one. The modulo terms can be calculated by using the derivations of the spin-frame in Appendix \ref{app_3_express}.
\end{remark}

\section{Dirac fields and energy identities}\label{S4}
\subsection{The massless Dirac equation}
The massless spin-$1/2$ Dirac field is a solution of the Weyl anti-neutrino equation
\begin{equation}\label{Or}
\nabla^{AA'}\phi_A = 0.
\end{equation}
This equation is conformal invariant, i.e, if $\phi_A$ is a solution of \eqref{Or} then $\widetilde{\phi} = \Omega^{-1}\phi = r\phi$
is a solution of the rescaled equation
\begin{equation}\label{Res}
\widetilde{\nabla}^{AA'}\widetilde{\phi}_A =0,
\end{equation}
where $\widetilde{\nabla}$ is the Levi-Civita connection for the rescaled metric $\widetilde{g}$.

By decomposing $\phi_A$ and $\widetilde{\phi}_A$ on the spin frames $\left\{ o_A,\iota_A\right\}$ and $\left\{ \widetilde{o}_A,\widetilde{\iota}_A\right\}$ we find that 
\begin{eqnarray*}
\widetilde{\phi}_A &=& \widetilde{\phi}_1\widetilde{o}_A - \widetilde{\phi}_0\widetilde{\iota}_A = \widetilde{\phi}_1 o_A - R\widetilde{\phi}_0\iota_A\cr
&=& r\phi_1 o_A - r\phi_0\iota_A.
\end{eqnarray*}
Hence
$$\widetilde{\phi}_0 = r^2\phi_0,\, \widetilde{\phi}_1 = r\phi_1.$$

In the recaled star-Kerr coordinates $(^*t,R,\theta,{^*}\varphi)$ the rescaled Weyl equation \eqref{Res} can be expressed using the rescaled Newman-Penrose tetrad and the associated spin coefficients as follows 
\begin{eqnarray*}
0 = \widetilde{\nabla}^{AA'}\widetilde{\phi}_A = &&\left( \widetilde{D}'\widetilde{\psi}_0 - \widetilde{\delta}\widetilde{\psi}_1 + (\widetilde{\mu} - \widetilde{\gamma})\widetilde{\psi}_0 + (\widetilde{\tau} - \widetilde{\beta})\widetilde{\phi}_1 \right) \bar{\widetilde{o}}^{A'} \cr
&&+ \left( \widetilde{D}\widetilde{\phi}_1 - \widetilde{\delta}'\widetilde{\phi}_0 + (\widetilde{\alpha} - \widetilde{\pi})\widetilde{\phi}_0 + (\widetilde{\varepsilon} - \widetilde{\rho})\widetilde{\phi}_1 \right)\bar{\widetilde{\iota}}^{A'}.\label{SpinorResWeyl}
\end{eqnarray*} 
This is equivalent to
\begin{align}\label{Res1}
\begin{cases}
\widetilde{D}'\widetilde{\phi}_0 - \widetilde{\delta}\hat{\phi}_1 + (\widetilde{\mu} - \widetilde{\gamma})\widetilde{\phi}_0 + (\widetilde{\tau} - \widetilde{\beta})\widetilde{\phi}_1 = 0,\cr
\widetilde{D}\widetilde{\phi}_1 - \widetilde{\delta}'\widetilde{\phi}_0 + (\widetilde{\alpha} - \widetilde{\pi})\widetilde{\phi}_0 + (\widetilde{\varepsilon} - \widetilde{\rho})\widetilde{\phi}_1 = 0, 
\end{cases}
\end{align} 
where $\widetilde{D},\, \widetilde{D}',\,\widetilde{\delta} $ and $\widetilde{\delta}'$ are the directional
derivatives along $\widetilde{l}^a,\, \widetilde{n}^a,\, \widetilde{m}^a$ and $\widetilde{\bar m}^a$ respectively.
In the rescaled Kerr spacetime $(\widetilde{M},\widetilde{g})$, we have the twelve values of the rescaled spin coefficients which are (see \cite{Xuan2020}):
\begin{equation*} \label{ResSpin-coffi1}
 \widetilde{\kappa}=\widetilde{\sigma}=\widetilde{\lambda}=\widetilde{\nu}=0, 
\end{equation*} 
\begin{equation*} \label{ResSpin-coffi2}
\widetilde{\tau}=-\frac{ia\sin\theta r}{\sqrt 2 \rho^2} , \, \widetilde{\pi}=\frac{ia\sin\theta r}{\sqrt 2 \bar p^2}, \, 
\widetilde{\rho} = -\frac{iar\cos\theta}{\bar p}\sqrt{\frac{\Delta}{2\rho^2}}, \,  \widetilde{\mu} = \left(R-\frac{1}{\bar p}\right)\sqrt{\frac{\Delta}{2\rho^2}},
\end{equation*}
\begin{equation*} \label{ResSpin-coffi3}
\widetilde{\varepsilon} = \frac{Mr^4 - a^2r^2(r\sin^2\theta + M\cos^2\theta)}{2\rho^2\sqrt{2\Delta\rho^2}},\, 
\widetilde{\alpha} = \frac{r}{\sqrt 2\bar{p}}\left(\frac{ia\sin\theta}{\bar p}-\frac{\cot\theta}{2}+\frac{a^2\sin\theta\cos\theta}{2\rho^2}\right),
\end{equation*} 
\begin{equation*} \label{ResSpin-coffi4}
\widetilde{\beta}=\frac {r}{\sqrt 2 p}\left(\frac{\cot\theta}{2}+\frac{a^2\sin\theta\cos\theta}{2\rho^2}\right),
\end{equation*} 
\begin{equation*} \label{ResSpin-coffi5}
\widetilde{\gamma}=\frac{Mr^2 - a^2(r\sin^2\theta + M\cos^2\theta)}{2\rho^2\sqrt{2\Delta \rho^2}} - \left(\frac{ia\cos\theta}{\rho^2}+R\right)\sqrt{\frac{\Delta}{2\rho^2}}.
\end{equation*} 

Therefore, we have the following scalar expression of the rescaled equation \eqref{Res1}:
\begin{equation}\label{Scalar}
\begin{gathered}
\sqrt{\frac{2}{\Delta\rho^2}}\left( (r^2+a^2)\partial_{^*t} + a\partial_{^*\varphi} + \frac{R^2\Delta}{2}\partial_R \right)\widetilde{\phi}_0 - \frac{r}{\sqrt{2}p}\left( ia\sin\theta\partial_{^*t}+\partial_{\theta}+\frac{i}{\sin\theta}\partial_{^*\varphi} \right)\widetilde{\phi}_1 \cr
+\left( \left( 2R - \frac{r}{\rho^2} \right)\sqrt{\frac{\Delta}{2\rho^2}} - \frac{Mr^2-a^2(r\sin^2\theta+M\cos^2\theta)}{2\rho^2\sqrt{2\Delta\rho^2}} \right)\widetilde{\phi}_0 \cr
-\frac{r}{\sqrt{2}p}\left( \frac{ia\sin\theta}{\bar{p}} + \frac{\cot\theta}{2} + \frac{a^2\sin\theta\cos\theta}{2\rho^2}\right)\widetilde{\phi}_1 = 0,\cr
-\sqrt{\frac{\Delta}{2\rho^2}}\partial_R\widetilde{\phi}_1 + \frac{r}{\sqrt{2}\bar{p}}\left( ia\sin\theta\partial_{^*t} - \partial_{\theta} + \frac{i}{\sin\theta}\partial_{^*\varphi}  \right)\widetilde{\phi}_0 + \frac{r}{\sqrt{2}\bar{p}} \left( -\frac{\cot\theta}{2} + \frac{a^2\sin\theta\cos\theta}{2\rho^2}\right)\widetilde{\phi}_0 \cr
+ \left( \frac{Mr^4 - a^2r^2(r\sin^2\theta+M\cos^2\theta)}{2\rho^2\sqrt{2\Delta\rho^2}} + \frac{iar\cos\theta}{\bar{p}}\sqrt{\frac{\Delta}{2\rho^2}} \right)\widetilde{\phi}_1 = 0.
\end{gathered}
\end{equation}

By the same way as above we can express the rescaled Weyl equation in the rescaled Kerr-star coordinates $(t^*,R,\theta,\varphi^*)$. First, we have the rescaled equation $\widehat{\nabla}^{AA'}\widehat{\phi}_A=0$. The recasled field $\widehat{\phi}_A= r\phi_A$ and $\widehat{o}_A = R o_A$ and $\widehat{\iota}_A = \iota_A$ lead to $\widehat{\phi}_0 = r\phi_0$ and $\widehat{\phi}_1 = r^2\phi_1$. The twelve values of the rescaled spin coefficients does not change
\begin{equation*} \label{ResSpin-coffi1}
\widehat{\kappa}=\widehat{\sigma}=\widehat{\lambda}=\widehat{\nu}=0, 
\end{equation*} 
\begin{equation*} \label{ResSpin-coffi2}
\widehat{\tau}=-\frac{ia\sin\theta r}{\sqrt 2 \rho^2} , \, \widehat{\pi}=\frac{ia\sin\theta r}{\sqrt 2 \bar p^2}, \, 
\widehat{\rho} = -\frac{iar\cos\theta}{\bar p}\sqrt{\frac{\Delta}{2\rho^2}}, \,  \widehat{\mu} = \left(R-\frac{1}{\bar p}\right)\sqrt{\frac{\Delta}{2\rho^2}},
\end{equation*}
\begin{equation*} \label{ResSpin-coffi3}
\widehat{\varepsilon} = \frac{Mr^4 - a^2r^2(r\sin^2\theta + M\cos^2\theta)}{2\rho^2\sqrt{2\Delta\rho^2}},\, 
\widehat{\alpha} = \frac{r}{\sqrt 2\bar{p}}\left(\frac{ia\sin\theta}{\bar p}-\frac{\cot\theta}{2}+\frac{a^2\sin\theta\cos\theta}{2\rho^2}\right),
\end{equation*} 
\begin{equation*} \label{ResSpin-coffi4}
\widehat{\beta}=\frac {r}{\sqrt 2 p}\left(\frac{\cot\theta}{2}+\frac{a^2\sin\theta\cos\theta}{2\rho^2}\right),
\end{equation*} 
\begin{equation*} \label{ResSpin-coffi5}
\widehat{\gamma}=\frac{Mr^2 - a^2(r\sin^2\theta + M\cos^2\theta)}{2\rho^2\sqrt{2\Delta \rho^2}} - \left(\frac{ia\cos\theta}{\rho^2}+R\right)\sqrt{\frac{\Delta}{2\rho^2}}.
\end{equation*} 
Therefore, in the rescaled Kerr-star coordinates $(t^*,R,\theta,\varphi^*)$ the rescaled equation $\widehat{\nabla}^{AA'}\widehat{\phi}_A=0$ has the following scalar expression
\begin{equation}\label{Scalar1}
\begin{gathered}
\sqrt{\frac{\Delta}{2\rho^2}}\partial_R \widehat{\phi}_0 - \frac{r}{\sqrt{2}p}\left( ia\sin\theta\partial_{t^*}+\partial_{\theta}+\frac{i}{\sin\theta}\partial_{\varphi^*} \right)\widehat{\phi}_1 \cr
+\left( \left( 2R - \frac{r}{\rho^2} \right)\sqrt{\frac{\Delta}{2\rho^2}} - \frac{Mr^2-a^2(r\sin^2\theta+M\cos^2\theta)}{2\rho^2\sqrt{2\Delta\rho^2}} \right)\widehat{\phi}_0 \cr
-\frac{r}{\sqrt{2}p}\left( \frac{ia\sin\theta}{\bar{p}} + \frac{\cot\theta}{2} + \frac{a^2\sin\theta\cos\theta}{2\rho^2}\right)\widehat{\phi}_1 = 0,\cr
\sqrt{\frac{2}{\Delta\rho^2}}\left( (r^2+a^2)\partial_{t^*} + a\partial_{\varphi^*} - \frac{R^2\Delta}{2}\partial_R \right)\widehat{\phi}_1 \cr
+ \frac{r}{\sqrt{2}\bar{p}}\left( ia\sin\theta\partial_{t^*} - \partial_{\theta} + \frac{i}{\sin\theta}\partial_{\varphi^*}  \right)\widehat{\phi}_0 + \frac{r}{\sqrt{2}\bar{p}} \left( -\frac{\cot\theta}{2} + \frac{a^2\sin\theta\cos\theta}{2\rho^2}\right)\widehat{\phi}_0 \cr
+ \left( \frac{Mr^4 - a^2r^2(r\sin^2\theta+M\cos^2\theta)}{2\rho^2\sqrt{2\Delta\rho^2}} + \frac{iar\cos\theta}{\bar{p}}\sqrt{\frac{\Delta}{2\rho^2}} \right)\widehat{\phi}_1 = 0.
\end{gathered}
\end{equation}

The Cauchy problem of the Weyl equation $\nabla^{AA'}\phi_A=0$ on the block $\mathcal{B}_I$ was solved in \cite{Ni2002}. In order to extend the solution to the conformal boundary $\mathfrak{H}^+\cup \scri^+$ we need to consider the Cauchy problem of the rescaled Weyl equation $\widetilde{\nabla}^{AA'}\widetilde{\phi}_A=0$ (resp. $\widehat{\nabla}^{AA'}\widehat{\phi}_A=0$) on $\bar{\mathcal{B}}_I$ minus a union $\mathcal{O}$ of arbitrary neighbourhoods of $i^+$ and $i_0$. This work can be done by the same way for the nonlinear Klein-Gordon equation on Kerr spacetime \cite[Section 4.2]{Ni2002'}. 
\begin{theorem}\label{cauchyproblem}(Cauchy problem)
The Cauchy problem for the rescaled Weyl equation $\widetilde{\nabla}^{AA'}\widetilde{\phi}_A=0$ (resp. $\widehat{\nabla}^{AA'}\widehat{\phi}_A=0$) in $\bar{\mathcal{B}}_I/\mathcal{O}$ ($\mathcal{O}$ is a union of the arbitrary neighbourhoods of $i^+$ and $i_0$) is well-posed, i.e, for any $\widetilde{\psi}_A \in \mathcal{C}_0^\infty(\Sigma_0, \mathbb{S}_A)$ (resp. $\widehat{\psi}_A \in \mathcal{C}_0^\infty(\Sigma_0, \mathbb{S}_A)$) there exists a unique $\widetilde{\phi}_A$ (resp. $\widehat{\phi}_A$) solution of $\widetilde{\nabla}^{AA'}{\widetilde{\phi}}_A = 0$ (resp. $\widehat{\nabla}^{AA'}\widehat{\phi}_A=0$) such that   
$$\widetilde{\phi}_A \in {\mathcal C}^\infty(\bar{\mathcal{B}}_I/\mathcal{O}, \mathbb{S}_A); \; \widetilde{\phi}_A|_{t=0} = \widetilde{\psi}_A.$$
(resp. $\widehat{\phi}_A \in {\mathcal C}^\infty(\bar{\mathcal{B}}_I/\mathcal{O}, \mathbb{S}_A); \; \widehat{\phi}_A|_{t=0} = \widehat{\psi}_A.$)
\end{theorem}
\begin{remark}
Since $\mathcal{O}$ is a union of the arbitrary neighbourhoods of $i^+$ and $i_0$ (we can choose $\mathcal{O}$ arbitrary small), the existence and uniqueness of the rescaled solution in ${\cal C}^{\infty}(\bar{\mathcal{B}}_I/\mathcal{O};\mathbb{S}_A)$ of the Cauchy problem of the rescaled Weyl equation allows us to extend smooth the rescaled solution $\widetilde{\phi}_A$ (resp. $\widehat{\phi}_A$) on the conformal boundary $\mathfrak{H}^+\cup \scri^+$.
\end{remark}
\subsection{Hyperboloid foliation and assumption on decay rate of field} 
We consider a hyperboloidal foliation $\left\{\mathcal{H}_{\tau}\right\}_{\tau\geq 0}$ of $\bar{\mathcal{B}}_I$ given by
$$\mathcal{H}_\tau = \left\{(\tau,h(r),\theta,\varphi)| \tau = t^* - h(r) = t+r_* - h(r)\right\},$$
where 
\begin{eqnarray}\label{timefunction}
h(r) &=& 2(r-r_+) + 4M \log \left( \frac{r}{r_+} \right) + \frac{3M^2(r_+-r)^2}{r_+ r^2} + 2M \arctan \frac{(C-1)M}{r}\cr
&&-2M\arctan \frac{(C-1)M}{r_+}\,\,\, (\hbox{for   }C\geq 1).
\end{eqnarray}
Lemma $2.21$ in \cite{ABM} shows that for $r\geq r_+$, $h(r)$ satisfies 
$$\lim_{r\to r_+}h(r)=0, \lim_{r\to \infty} \frac{h(r)}{r} =2,\, \lim_{r\to r_+}\partial_r h(r) = 1,\, \partial_rh(r)\geq 0.$$

Since $\lim_{r\to \infty}\dfrac{h(r)}{r} = 2$, for $r$ large enough the hyperboloid hypersurface $\mathcal{H}_\tau$ is asymptotic to the level set of $^*t$.

Moreover, we have that
$$\partial_r h(r) = 2+\frac{4M}{r} + \frac{6M^2(r-r_+)}{r^3} - \frac{2(C-1)M^2}{(C-1)^2M^2+r^2}.$$
This leads to
\begin{equation}\label{CauCon}
0\leq \frac{a^2+r^2}{\Delta}\left( 1 - \sqrt{1-\frac{a^2\Delta}{(r^2+a^2)^2}} \right) < \partial_r h(r) < \frac{a^2+r^2}{\Delta}\left( 1 + \sqrt{1-\frac{a^2\Delta}{(r^2+a^2)^2}} \right).
\end{equation}
The inequality \eqref{CauCon} guarantees $\mathcal{H}_\tau$ is a Cauchy hypersurface (for more details see the proof of \cite[Lemma 2.1]{ABM}). 

There are some results about the pointwise decays (also called Price's law) of the Dirac fields on the Schwarzschild and Kerr spacetimes \cite{FiSm1,FiSm2,SmolXi,Ma2021}. 
In particular, in Schwarzschild spacetime Smoller and Xie \cite{SmolXi} use Chandrasekhar's separation of variables whereby the Dirac equations
split into two sets of wave equations, then show that the wave decays as $t^{-2\lambda}$, where $\lambda=1,2...$ is the angular momentum. On the other hand, Ma in \cite{Ma2021} establishes the decays by transforming the Dirac equation to the Teukolsky spin wave equations and then use the vector field method which consists Morawertz estimates and $r^p$-theory to obtain the decays in both time and spatial directions for these equations in Schwarzschild spacetime (see also \cite{ABM} for the linearized gravity fields). The decay results obtained for the Dirac field in \cite{Ma2021} (in detail the decay rates of the components are $v^{-\alpha}\tau^{-\beta}$) is improved the one in \cite{SmolXi}. 
Moreover, the time decay for the Dirac field on Kerr spacetime can be found in the works Finster et al. \cite{FiSm1,FiSm2}, where the the decay rate is $t^{-5/6}$. 

We notice that we need not an improved decay rate to construct a conformal scattering theory for the Dirac field. We need only that the field decays in both time and spatial directions with the suffit spatial decay rate. 
Therefore, in this work we assume that in abstract Kerr spacetime the Dirac field $\phi_A$ decays both in the time and spatial directions and the decay rate of the components $\phi_1$ and $\phi_0$ as follows 
\begin{equation}\label{pointwise}
|\phi_1|\lesssim t^{-1-\alpha}(t-r)^{-2-\beta},\, |\phi_0| \lesssim t^{-2-\alpha'}(t-r)^{-1-\beta'}
\end{equation}
where the constants $\alpha,\beta, \alpha',\beta' \geq 0$. This condition is valid for $a=0$, i.e, Schwarzschild spacetime (see \cite[Theorem 1.6]{Ma2021}) and for $0<a \ll M$, i.e, very slowly Kerr spacetime (see Appendix \ref{verySlowKer}).

\subsection{Energy fluxes and functional space of initial data}
The rescaled Weyl equation \eqref{Res} admits a null and future time orientation and conserved current
\begin{eqnarray}\label{CurrentConserved}
\widetilde{J}^a = \widetilde{\phi}^A\bar{\widetilde\phi}^{A'} &=& \widetilde{\varepsilon}^{AB}\widetilde{\phi}_B\widetilde{\varepsilon}^{A'B'}\bar{\widetilde{\phi}}_{B'} =\widetilde{\varepsilon}^{AB}\widetilde{\varepsilon}^{A'B'}(\widetilde{\phi}_1\widetilde{o}_B - \widetilde{\phi}_0\widetilde{\iota}_B)(\bar{\widetilde{\phi}}_1\bar{\widetilde{o}}_{B'} - \bar{\widetilde{\phi}}_0\bar{\widetilde{\iota}}_{B'})\cr
&=&\widetilde{\varepsilon}^{AB}\widetilde{\varepsilon}^{A'B'} (\vert\widetilde{\phi}_1\vert^2 \widetilde{o}_B\bar{\widetilde{o}}_{B'} - \widetilde{\phi}_1\bar{\widetilde{\phi}}_0\widetilde{o}_B\bar{\widetilde{\iota}}_{B'} - \widetilde{\phi}_0\bar{\widetilde{\phi}}_1\widetilde{\iota}_B\bar{\widetilde{o}}_{B'} + \vert\widetilde{\phi}_0\vert^2\widetilde{\iota}_B\bar{\widetilde{\iota}}_{B'} ) \cr
&=& \vert\widetilde{\phi}_1\vert^2\widetilde{l}^a\partial_a + \vert\widetilde{\phi}_0\vert^2\widetilde{n}^a\partial_a - \widetilde{\phi}_1\bar{\widetilde\phi}_0\widetilde{m}^a\partial_a - \widetilde{\phi}_0\bar{\widetilde{\phi}}_1\bar{\widetilde{m}}^a\partial_a.
\end{eqnarray}
Its Hodge dual is given as
\begin{eqnarray}\label{CurrentConserved1}
*\widetilde{J}_a\d x^a &=& \widetilde{J}^a \partial_a \hook \mathrm{dVol}\cr
&=& \left( \vert\widetilde{\phi}_1\vert^2\widetilde{l}^a\partial_a + \vert\widetilde{\phi}_0\vert^2\widetilde{n}^a\partial_a - \widetilde{\phi}_1\bar{\widetilde\phi}_0\widetilde{m}^a\partial_a - \bar{\widetilde{\phi}}_1\widetilde{\phi}_0\bar{\widetilde{m}}^a\partial_a  \right)\hook \mathrm{dVol}\cr
&=& \left( \vert\widetilde{\phi}_1\vert^2\widetilde{l}^a\partial_a + \vert\widetilde{\phi}_0\vert^2\widetilde{n}^a\partial_a - \widetilde{\phi}_1\bar{\widetilde\phi}_0\widetilde{m}^a\partial_a - \bar{\widetilde{\phi}}_1\widetilde{\phi}_0\bar{\widetilde{m}}^a\partial_a  \right)\hook i\widetilde{l}\wedge\widetilde{n}\wedge\widetilde{m}\wedge\bar{\widetilde{m}} \cr
&=& -i\widetilde{l}\wedge\widetilde{m}\wedge\bar{\widetilde{m}}\vert\widetilde{\phi}_1\vert^2 + i\widetilde{n}\wedge\widetilde{m}\wedge\bar{\widetilde{m}}\vert\widetilde{\phi}_0\vert^2 + i\widetilde{l}\wedge\widetilde{n}\wedge\widetilde{m}\widetilde{\phi}_1\bar{\widetilde{\phi}}_0 - i\widetilde{l}\wedge\widetilde{n}\wedge\bar{\widetilde{m}}\bar{\widetilde\phi}_1\widetilde{\phi}_0.
\end{eqnarray}
Since $\widetilde{\nabla}^{AA'}\widetilde{\phi}_A = 0$, we have the following conservation law 
\begin{equation}
\widetilde{\nabla}^a \widetilde{J}_a =0.
\end{equation}

Let ${S}$ be a boundary of a bounded open set $\Omega$ and has outgoing
orientation, the energy flux on a oriented hypersurface $\cal{S}$ is defined as
$${\cal E}_{\mathcal{S}}(\widetilde{\phi}_A) = -4\int_{\cal S}*\widetilde{J}_a\d x^a.$$ 
Using Stokes theorem, we have
\begin{equation}\label{Stokesformula}
-4\int_{\mathcal{S}}*\widetilde{J}_a\d x^a = \int_{\Omega}(\widetilde{\nabla}_a\widetilde{J}^a)\mathrm{{dVol}}_{\widetilde g}.
\end{equation}
Therefore
\begin{equation}\label{energy}
{\cal E}_{\mathcal{S}}(\widetilde{\phi}_A) = \int_{\mathcal{S}} \widetilde{J}_a\widetilde{N}^a\widetilde{L}\hook \mathrm{{dVol}}_{\widetilde g}|_{\mathcal{S}},
\end{equation}
where $\widetilde{L}$ is a transverse vector to $\mathcal{S}$ and $\widetilde{N}$ is the normal vector field to $\mathcal{S}$ such that $\widetilde{L}^a\widetilde{N}_a=1$.

We cut-off the compactifiation domain $\bar{\mathcal{B}}_I$ by the hyperboloid hypersurface $\mathcal{H}_T,\, T>0$ of the foliation $\left\{ \mathcal{H}_\tau\right\}_{\tau\geq 0}$. We consider domain $\Omega\subset \bar{\mathcal{B}}_I$ which is closed by the hypersufaces $\Sigma_0$, $\mathfrak{H}^+_T$, $\scri^+_T$ and $\Sigma_T$ for $T>0$ large enoguh. We can see that $\Sigma_T$ tends to $i^+$ as $T$ tends to infinity.

\begin{figure}[H]
\begin{center}
\includegraphics[scale=0.7]{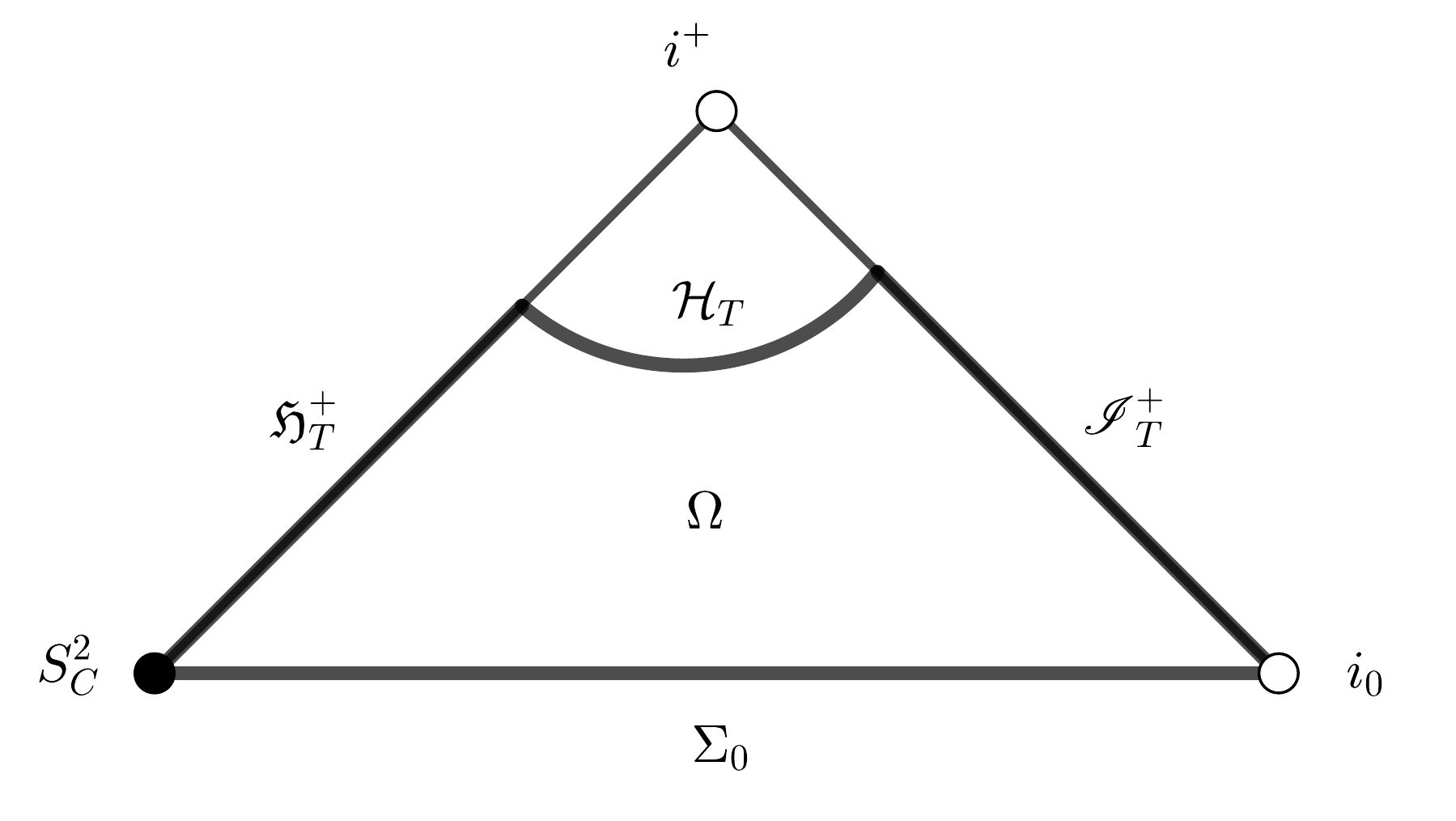}
\caption{The hyperboloid hypersurface $\mathcal{H}_T$ and domain $\Omega \subset \bar{\mathcal{B}}_I$.}
\end{center}
\end{figure}

\begin{lem}\label{EquivalentSimplerEnergy}
Consider the smooth and compactly supported initial data on $\Sigma_0$,
the energy fluxes of the Dirac field through Cauchy hypersurface $\Sigma_0$, the conformal boundaries $\mathfrak{H}_T^+$ and $\scri_T^+$ and $\mathcal{H}_T$ have the following simpler equivalent expressions
$$\mathcal{E}_{\Sigma_0}(\widetilde{\phi}_A) \simeq \frac{1}{\sqrt{2}} \int_{\Sigma_0} (R^2|\widetilde{\phi}_0|^2 + |\widetilde{\phi}_1|^2) \d r_*\d^2\omega,$$
$$\,\mathcal{E}_{\scri_T^+}(\widetilde{\phi}_A) \simeq \frac{1}{\sqrt{2}} \int_{\scri^+_T} |\widetilde{\phi}_1|_{\scri^+}|^2 \d {^*t}\d^2\omega,\, \mathcal{E}_{\mathfrak{H}_T^+}(\widetilde{\phi}_A) \simeq \frac{1}{\sqrt{2}}\int_{\mathfrak{H}_T^+}\vert\sqrt{\Delta}\widetilde{\phi}_0|_{\mathfrak{H}^+}\vert^2 \d {t^*}\d^2\omega,$$
$$\mathcal{E}_{\mathcal{H}_T}(\widetilde{\phi}_A) \simeq \int_{\mathcal{H}_T} \left( \vert\widetilde{\phi}_1\vert^2 + \vert\widetilde{\phi}_0\vert^2 \right) R^2 \d r_*\d^2\omega,$$
for $r$ large enough.
\end{lem}
\begin{proof}
On $\Sigma_0$, we take
$$\widetilde{L}= a_1\partial_{^*t} + a_2\partial_{^*\varphi}, \, \widetilde{N}=\partial_{^*t}=\partial_t,$$
where the functional coefficients $a_1,\, a_2$ are positive and satisfying that $g(\widetilde{L}^a,\widetilde{N}^a) =1$, i.e, 
$$a_1 R^2\left( 1-\frac{2Mr}{\rho^2} \right) + a_2 \frac{2MaR\sin^2\theta}{\rho^2} =1.$$
Therefore,
$$\widetilde{L}^a \hook \mathrm{dVol}_{\widetilde{g}}|_{\Sigma_0}  = a_1\partial_{^*t} \hook \mathrm{dVol}_{\widetilde{g}} \simeq \d r_*\wedge \d^2\omega.$$ 
On the other hand, using \eqref{resNew2} and \eqref{CurrentConserved} we have
\begin{eqnarray*}
\widetilde{g}(\widetilde{J}^a,\widetilde{N}^a) &=& \widetilde{g}\left( \vert\widetilde{\phi}_1\vert^2\widetilde{l}^a + \vert\widetilde{\phi}_0\vert^2\widetilde{n}^a - \widetilde{\phi}_1\bar{\widetilde\phi}_0\widetilde{m}^a - \widetilde{\phi}_0\bar{\widetilde{\phi}}_1\bar{\widetilde{m}}^a , \partial_{^*t} \right)\cr 
&=& \sqrt{\frac{\Delta}{2\rho^2}}(\vert\widetilde{\phi}_1\vert^2 + R^2\vert\widetilde{\phi}_0\vert^2) + \frac{iaR\sin\theta}{p\sqrt{2}} \widetilde{\phi}_1\bar{\widetilde\phi}_0 - \frac{iaR\sin\theta}{\bar{p}\sqrt{2}} \widetilde{\phi}_0\bar{\widetilde\phi}_1\cr
&\simeq& \frac{1}{\sqrt{2}}(\vert\widetilde{\phi}_1\vert^2 + R^2\vert\widetilde{\phi}_0\vert^2),
\end{eqnarray*}
for $r>r_+$ ($r_+$ large enough). Therefore, we obtain that
$$\mathcal{E}_{\Sigma_0}(\widetilde{\phi}_A) \simeq \frac{1}{\sqrt{2}} \int_{\Sigma_0} (R^2|\widetilde{\phi}_0|^2 + |\widetilde{\phi}_1|^2) \d r_*\d^2\omega.$$

On $\scri^+$, we take $\widetilde{N}=\partial_{^*t}$ on $\scri^+$ and $\widetilde{L}_{\scri^+}= -\partial_R$ in the rescaled star-Kerr coordinates $({^*t},R,\theta,{^*\varphi})$:
$$\widetilde{L}_{\scri^+}= \frac{2r^2}{\Delta}l|_{\scri^+},$$
where $l^a\partial_a$ given by \eqref{New2}. The energy flux through $\scri^+_T$ is
\begin{eqnarray*}
\mathcal{E}_{\scri^+_T}(\widetilde{\phi}_A) &\simeq& \int_{\scri^+} \frac{1}{\sqrt{2}}\left(\vert\widetilde{\phi}_1\vert^2 + R^2\vert\widetilde{\phi}_0\vert^2 \right)|_{\scri^+} \d {^*t}\d^2\omega\cr
&\simeq& \frac{1}{\sqrt{2}}\int_{\scri^+}\vert\widetilde{\phi}_1|_{\scri^+}\vert^2 \d {^*t}\d^2\omega.
\end{eqnarray*}

By the same way as above, on $\mathfrak{H}^+$, we take $\widehat{N}=\partial_{t^*}$ and $\widehat{L}_{\mathfrak{H}^+} = \partial_R$ in the rescaled Kerr-star coordinates $({t^*},R,\theta, \varphi^*)$ with the rescaled metric $\widehat{g}$:
$$\widehat{L}_{\mathfrak{H}^+} = \frac{2r^2}{\Delta}n|_{\mathfrak{H}^+},$$
where $n^a\partial_a$ given by \eqref{New2}. The energy flux through $\mathfrak{H}^+$ is
\begin{eqnarray}\label{eqonH}
\mathcal{E}_{\mathfrak{H}^+_T}(\widehat{\phi}_A) &\simeq& \int_{\mathfrak{H}_T^+}\left[ \sqrt{\frac{\Delta}{2\rho^2}}\left(R^2\vert\widehat{\phi}_1\vert^2 + \vert\widehat{\phi}_0\vert^2 \right)\right]|_{\mathfrak{H}^+} \d {t^*}\d^2\omega\cr
&\simeq& \frac{1}{\sqrt{2}}\int_{\mathfrak{H}_T^+}\vert\sqrt{\Delta}\widehat{\phi}_0|_{\mathfrak{H}^+}\vert^2 \d {t^*}\d^2\omega.
\end{eqnarray}
The last equivalence is due to the contraction $(\sqrt{\Delta}\widehat{\phi}_1)|_{\mathfrak{H}^+}=0$. Moreover, $\widehat{\phi}_1$ depends to $\sqrt{\Delta}\widetilde{\phi}_0$ on $\mathfrak{H}^+$ (see equation \eqref{Phi1eq}). The equivalent equality \eqref{eqonH} leads to 
$$\mathcal{E}_{\mathfrak{H}^+_T}(\widetilde{\phi}_A) \simeq \frac{1}{\sqrt{2}}\int_{\mathfrak{H}_T^+}\vert\sqrt{\Delta}\widetilde{\phi}_0|_{\mathfrak{H}^+}\vert^2 \d {t^*}\d^2\omega.$$

On $\mathcal{H}_T$ we have $T = t^* - h(r) = t+r_* - h(r)$ then a co-normal to $\mathcal{H}_T$ is given by
$$\widetilde{N}_a\d x^a = \d {^*t} + \left(2 - \frac{\Delta}{r^2+a^2}\partial_r h\right) \d r_*.$$
Hence
$$\d {^*t} = \frac{r^2(r^2+a^2)}{\Delta} \left( 2-\frac{\Delta}{r^2+a^2}\partial_rh \right) \d R = r^2\alpha \d R,$$
where $\alpha = \frac{r^2+a^2}{\Delta} \left( 2-\frac{\Delta}{r^2+a^2}\partial_rh \right)$.

Therefore, the rescaled dual Newman-Penrose tetrad \eqref{dualNew3} restricted on ${\cal H}_T$ is
\begin{eqnarray*}
\widetilde{l}_a \d x^a &=& \sqrt{\frac{\Delta}{2\rho^2}} \left( \d {^*t} - a\sin^2\theta \d {^*\varphi} \right),\cr
\widetilde{n}_a|_{{\cal H}_T} \d x^a &=& \sqrt{\frac{\Delta}{2\rho^2}} \left( \left( R^2 - \frac{2\rho^2R^2}{\Delta\alpha} \right) \d {^*t} - aR^2\sin^2\theta \d {^*\varphi} \right),\cr
\widetilde{m}_a \d x^a &=& \frac{1}{p\sqrt 2} \left( iaR\sin\theta \d {^*t} - R\rho^2 \d \theta - iR(r^2+a^2)\sin\theta \d {^*\varphi} \right).
\end{eqnarray*}
These equalities yield that
\begin{eqnarray*}
\widetilde{l}\wedge\widetilde{m}\wedge\bar{\widetilde{m}} &=& \sqrt{\frac{\Delta}{2\rho^2}}\frac{1}{2\rho^2} (-2iR^2\rho^2(r^2+a^2)\sin\theta \d{^*t}\wedge \d\theta \wedge \d{^*\varphi} \cr
&& \hspace{3cm}+ 2ia^2R^2\rho^2\sin^3\theta \d{^*t}\wedge \d\theta \wedge \d{^*\varphi} )\cr
&=&\sqrt{\frac{\Delta\rho^2}{2}}iR^2\d{^*t}\d^2\omega,
\end{eqnarray*}
\begin{eqnarray*}
\widetilde{n}|_{{\cal H}_T}\wedge\widetilde{m}\wedge\bar{\widetilde{m}} &=& \sqrt{\frac{\Delta}{2\rho^2}}\frac{1}{2\rho^2} \left( -2iR^4\rho^2\left( 1 - \frac{2\rho^2}{\Delta\alpha}\right)(r^2+a^2)\sin\theta\d{^*t}\wedge \d\theta \wedge \d{^*\varphi} \right.\cr
&&\hspace{3cm} \left. +2ia^2R^4\rho^2\sin^3\theta \d{^*t}\wedge\d\theta\wedge\d{^*\varphi} \right) \\
&=& -\sqrt{\frac{\Delta\rho^2}{2}}i\left( 1 - \frac{2}{2-\frac{\Delta}{r^2+a^2}\partial_rh} \right)R^4\d{^*t}\d^2\omega,
\end{eqnarray*}
and
\begin{eqnarray*}
\widetilde{l}\wedge\widetilde{n}|_{{\cal H}_T}\wedge \widetilde{m}&=& \frac{\Delta}{2\rho^2}\frac{1}{p\sqrt{2}}\left(-aR^3\rho^2\sin^2\theta \d{^*t}\wedge\d\theta\wedge\d{^*\varphi}\right.\cr
&&\hspace{3cm} \left. +aR\rho^2\sin^2\theta\left( R^2 - \frac{2R^2\rho^2}{\Delta\alpha} \right) \d{^*t}\wedge\d\theta\wedge\d{^*\varphi}\right)\cr
&=&-\frac{a\Delta R^3\bar{p}\sin\theta}{\sqrt{2}\Delta\alpha}\d{^*t}\d^2\omega,
\end{eqnarray*}
similarly
\begin{eqnarray*}
\widetilde{l}\wedge\widetilde{n}|_{{\cal H}_T}\wedge \bar{\widetilde{m}}&=& \frac{a\Delta R^3p\sin\theta}{\Delta\alpha}\d{^*t}\d^2\omega.
\end{eqnarray*}
Hence, the Hodge dual \eqref{CurrentConserved1} restricted on ${\cal H}_T$ is
\begin{eqnarray*}
-\omega|_{{\cal H}_T}&=& \sqrt{\frac{\Delta\rho^2}{2}}R^2\vert\widetilde{\phi}_1\vert^2\d{^*t}\d^2\omega + \sqrt{\frac{\Delta\rho^2}{2}}\left( \frac{2}{2-\frac{\Delta}{r^2+a^2}\partial_rh} - 1 \right)R^4\vert\widetilde{\phi}_0\vert^2 \d{^*t}\d^2\omega\cr 
&&-\frac{ia\bar{p}\Delta R^3\sin\theta}{\Delta\alpha}\widetilde{\phi}_1\bar{\widetilde{\phi}}_0\d{^*t}\d^2\omega + \frac{iap\Delta R^3\sin\theta}{\sqrt{2}\Delta\alpha}\bar{\widetilde{\phi}}_1 \widetilde{\phi}_0\d{^*t}\d^2\omega\cr
&=& \sqrt{\frac{\Delta\rho^2R^4}{2}} \left( \vert\widetilde{\phi}_1\vert^2 + \left( \frac{2}{2-\frac{\Delta}{r^2+a^2}\partial_rh} - 1 \right)R^2\vert\widetilde{\phi}_0\vert^2 \right)\d{^*t}\d^2\omega\cr
&&- \left( \frac{ia\bar{p}R^3\sin\theta}{\sqrt{2}\alpha}\widetilde{\phi}_1\bar{\widetilde{\phi}}_0 - \frac{iapR^3\sin\theta}{\sqrt{2}\alpha}\bar{\widetilde{\phi}}_1\widetilde{\phi}_0 \right)\d{^*t}\d^2\omega.
\end{eqnarray*}
We have that
\begin{eqnarray*}
&&\sqrt{\frac{\Delta\rho^2R^4}{2}} \left( 1 + \left( \frac{2}{2-\frac{\Delta}{r^2+a^2}\partial_rh} - 1 \right)R^2 \right)\cr
&\geq& 2\frac{|p|R^3\sqrt{\Delta}}{\sqrt{2}} \sqrt{\frac{2}{2-\frac{\Delta}{r^2+a^2}\partial_rh} - 1} \hbox{  and  } 2\frac{|\bar{p}|R^3\sqrt{\Delta}}{\sqrt{2}} \sqrt{\frac{2}{2-\frac{\Delta}{r^2+a^2}\partial_rh} - 1}. 
\end{eqnarray*}
We will show that
$$\sqrt{\frac{2}{2-\frac{\Delta}{r^2+a^2}\partial_rh} - 1} > \frac{1}{\alpha\sqrt{\Delta}}.$$
Putting $X= \frac{\Delta}{r^2+a^2}\partial_rh$, the above inequality is equivalent to
$$\sqrt{\frac{X}{2-X}} > \frac{a\sqrt{\Delta}}{(r^2+a^2)(2-X)}$$
and then 
$$2X-X^2 > \frac{a^2\Delta}{(r^2+a^2)^2}.$$
The last inequality holds by \eqref{CauCon}. Therefore, we can obtain two side estimates of $-\omega|_{{\cal H}_T}$ and give the energy flux through the hypersurface $\mathcal{H}_T$ as
\begin{eqnarray*}
\mathcal{E}_{\mathcal{H}_T}(\widetilde{\phi}_A) &\simeq& \int_{\cal{H}_T}\sqrt{\frac{\Delta\rho^2R^4}{2}} \left( \vert\widetilde{\phi}_1\vert^2 + \left( \frac{2}{2-\frac{\Delta}{r^2+a^2}\partial_rh} - 1 \right)R^2\vert\widetilde{\phi}_0\vert^2 \right)\d{^*t}\d^2\omega\cr
&\simeq&\int_{\cal{H}_T}\sqrt{\frac{\Delta\rho^2R^4}{2}} \left( \vert\widetilde{\phi}_1\vert^2 + \left( \frac{2}{2-\frac{\Delta}{r^2+a^2}\partial_rh} - 1 \right)R^2\vert\widetilde{\phi}_0\vert^2 \right)\left(2-\frac{\Delta}{r^2+a^2}\partial_rh \right)\d r_*\d^2\omega.
\end{eqnarray*}
For $r$ large enough, by \eqref{timefunction} we have
$$2-\frac{\Delta}{r^2+a^2}\partial_rh \simeq 2- \frac{\Delta}{r^2+a^2}\left( 2+ \frac{4M}{r}\right) \simeq R^2. $$
Therefore
$$\mathcal{E}_{\mathcal{H}_T}(\widetilde{\phi}_A) \simeq \int_{\cal{H}_T}\left( \vert\widetilde{\phi}_1\vert^2 + \vert\widetilde{\phi}_0\vert^2 \right)R^2\d r_*\d^2\omega.$$

\end{proof}
\begin{proposition}
Consider the smooth and compactly supported initial data on $\Sigma_0$, we can define the energy fluxes of the rescaled solution $\widetilde{\phi}_A$ across the null conformal boundary $\mathfrak{H}^+\cup \scri^+$ by 
$${\mathcal{E}}_{\scri^+}(\widetilde{\phi}_A) + {\mathcal{E}}_{\mathfrak{H}^+}(\widetilde{\phi}_A) := \lim_{T\rightarrow\infty}\mathcal{E}_{\scri_T^+}(\widetilde{\phi}_A) + {\mathcal{E}}_{\mathfrak{H}_T^+}(\widetilde{\phi}_A).$$
Moreover, we have
$$\mathcal{E}_{\scri^+}(\widetilde{\phi}_A) + \mathcal{E}_{\mathfrak{H}^+}(\widetilde{\phi}) \leq \mathcal{E}_{\Sigma_0}(\widetilde{\phi}_A),$$
where the equality holds if and only if $\lim_{T\rightarrow \infty}\mathcal{E}_{\mathcal{H}_T}(\widetilde{\phi}_A) = 0$.\end{proposition} 
\begin{proof} 
Intergrating the conservation law $\widetilde{\nabla}^a \widetilde{J}_a =0$ on the domain $\Omega$ and by using the Stokes's formula \eqref{Stokesformula}, we get an exact energy identity between the hypersurfaces $\Sigma_0, \, \mathfrak{H}_T^+, \, \mathcal{M}_T, \, \mathcal{N}_T$ and $\scri_T^+$ (for $T>0$) as follows 
\begin{equation}\label{energyidentity}
\mathcal{E}_{\Sigma_0}(\widetilde{\phi}_A)= \mathcal{E}_{\scri_T^+}(\widetilde{\phi}_A) + \mathcal{E}_{\mathfrak{H}_T^+}(\widetilde{\phi}_A) + \mathcal{E}_{\mathcal{H}_T}(\widetilde{\phi}_A).
\end{equation}
Therefore, the energy fluxes through $\scri_T^+$ and $\mathfrak{H}_T^+$ are non negative increasing functions of $T$ and their sum is bounded by $\mathcal{E}_{\Sigma_0}(\widetilde{\phi}_A)$. Hence, the limit of $\mathcal{E}_{\scri_T^+}(\widetilde{\phi}_A) + \mathcal{E}_{\mathfrak{H}_T^+}(\widetilde{\phi}_A)$ exists and the following sum is well defined
\begin{equation}\label{limitenergy}
\mathcal{E}_{\scri^+}(\widetilde{\phi}_A) + \mathcal{E}_{\mathfrak{H}^+}(\widetilde{\phi}_A) := \lim_{T\rightarrow\infty}\mathcal{E}_{\scri_T^+}(\widetilde{\phi}_A) + \mathcal{E}_{\mathfrak{H}_T^+}(\widetilde{\phi}_A) = \mathcal{E}_{\Sigma_0}(\widetilde{\phi}_A) - \lim_{T\rightarrow\infty}\mathcal{E}_{\mathcal{H}_T}(\widetilde{\phi}_A).
\end{equation}
The proposition now holds from the above identity.
\end{proof}

\subsection{Energy identity up to $i^+$ and trace operator}
\begin{theorem}\label{egalite_energies}
Let $\widetilde{\phi}_A$ be a solution to the rescaled equation \eqref{Res}, with the smooth and compactly supported initial data on $\Sigma_0$ satisfying some bounded conditions on the hyperboloid hypersurface $\Sigma_0$ such that the pointwise decay assumption \eqref{pointwise} holds.
The energies of $\widetilde{\phi}_A$ through the hypersurfaces $\mathcal{H}_T$ tend to zero as $T$ tends to infinity
$$\lim_{T\rightarrow \infty} \mathcal{E}_{\mathcal{H}_T}(\widetilde{\phi}_A) = 0.$$
As a consequence, the equality of the energies holds true
\begin{equation}\label{energyequality}
\mathcal{E}_{\Sigma_0}(\widetilde{\phi}_A) = \mathcal{E}_{\mathfrak{H}^+}(\widetilde{\phi}_A) + \mathcal{E}_{\scri^+}(\widetilde{\phi}_A).
\end{equation}
\end{theorem}
\begin{proof}
Let $r_{FH}$ be a large constant such that for $r\geq r_{FH}$ we have the equivalent energies as in Lemma \ref{EquivalentSimplerEnergy} and $h(r)= (2+\varepsilon) r$, where $\varepsilon \to 0$ as $r\to\infty$. Our proof is separated in three domain $r_+\leq r \leq r_{NH}$, $r_{NH}\leq r \leq r_{FH}$ and $r_{FH}\leq r$.

In the domain $r_{FH}\leq r$, from Lemma \ref{EquivalentSimplerEnergy}, the pointwise decay \eqref{pointwise} of the Dirac's components and $h(r) = (2+\varepsilon)r$ for $r$ large enough we have that
\begin{eqnarray*}
\mathcal{E}_{\mathcal{H}_T}(\widetilde{\phi}_A) &\simeq& \int_{\mathcal{H}_T} \left( \vert\widetilde{\phi}_1\vert^2 + \vert\widetilde{\phi}_0\vert^2 \right) R^2 \d r_* \d^2\omega\cr
&\simeq& \int_{\mathcal{H}_T}\left( \vert \phi_1\vert^2 + \vert r\phi_0\vert^2 \right) \d r_* \d^2\omega\cr
&\lesssim&\int_{r_{FH}}^\infty\int_{\mathbb{S}^2} t^{-2-2\alpha}(t-r)^{-4-2\beta} + t^{-4-2\alpha'}r^2(t-r)^{-2-2\beta'} \d r\d^2\omega\cr
&\lesssim& 2\pi\int_{r_{FH}}^\infty (T- r + h(r))^{-2-2\alpha}(T+h(r)-2r)^{-4-2\beta} \d r\cr
&& + 2\pi \int_{r_{FH}}^\infty (T-r + h(r))^{-4-2\alpha'}r^2(T+h(r)-2r)^{-2-2\beta} \d r\cr
&\lesssim& 2\pi\int_{r_{FH}}^\infty (T+ r)^{-2-2\alpha}T^{-4-2\beta} \d r + 2\pi \int_{r_{FH}}^\infty (T+r)^{-4-2\alpha'}r^2T^{-2-2\beta} \d r\cr
&\lesssim& 2\pi T^{-4-2\beta}(T+ r_{FH})^{-1-2\alpha} + 2\pi T^{-2-2\beta} \int_{r_{FH}}^\infty (T+r)^{-4-2\alpha'}r^2 \d r\cr
&\lesssim& 2\pi T^{-4-2\beta}(T+ r_{FH})^{-1-2\alpha} + 2\pi T^{-2-2\beta} (r_{FH})^{-1-2\alpha'} \longrightarrow 0
\end{eqnarray*}
as $T\to \infty$. 

In the domain $r_{NH} \leq r \leq r_{FH}$ we can consider that $h(r) \simeq c(r)r$ $(1\leq c(r) <2)$ and we have
\begin{eqnarray*}
\mathcal{E}_{\mathcal{H}_T}(\widetilde{\phi}_A) &\simeq& \int_{r_{NH}}^{r_{FH}}\int_{\mathbb{S}^2} t^{-2-2\alpha}(t-r)^{-4-2\beta} + t^{-4-2\alpha'}r^2(t-r)^{-2-2\beta'} \d r\d^2\omega\cr
&\lesssim& 2\pi\int_{r_{NH}}^{r_{FH}} (T- r + h(r))^{-2-2\alpha}(T+h(r)-2r)^{-4-2\beta} \d r\cr
&& + 2\pi \int_{r_{NH}}^{r_{FH}} (T-r + h(r))^{-4-2\alpha'}r^2(T+h(r)-2r)^{-2-2\beta} \d r\cr
&\lesssim& 2\pi\int_{r_{NH}}^{r_{FH}} T^{-2-2\alpha}(T-r)^{-4-2\beta} \d r + 2\pi \int_{r_{NH}}^{r_{FH}} T^{-4-2\alpha'}r_*^2(T-r)^{-2-2\beta} \d r\cr
&\lesssim& 2\pi T^{-2-2\alpha}(T - r_{FH})^{-3-2\beta} + 2\pi T^{-4-2\alpha'} \int_{r_{NH}}^{r_{FH}} r^2(T-r)^{-2-2\beta} \d r\cr
&\lesssim&  2\pi T^{-2-2\alpha}(T - r_{FH})^{-3-2\beta} + 2\pi T^{-4-2\alpha'} (r_{FH}-r_{NH})  \longrightarrow 0
\end{eqnarray*}
as $T\to \infty$.

In the domain $r_+\leq r\leq r_{FH}$ we obtain clearly that the limit of $\mathcal{E}_{\mathcal{H}_T}(\widetilde{\phi}_A)$ is also zero by the same way as above. Our proof is completed.
\end{proof}

A direct consequence of the energy equality \eqref{energyequality} is that we can define the trace operator on the null conformal boundaries $\mathfrak{H}^+\cup \scri^+$. 
\begin{defn}
The trace operator ${\mathcal T}^+: \mathcal{C}_0^{\infty}(\Sigma_0,\mathbb{S}_A) \to \mathcal{C}_0^\infty(\mathfrak{H}^+,\mathbb{C})\times \mathcal{C}_0^{\infty}(\scri^+,\mathbb{C})$ is given by
\begin{align*}
\mathcal{T}^+: \mathcal{C}_0^{\infty}(\Sigma_0,\mathbb{S}_A) &\longrightarrow \mathcal{C}_0^\infty(\mathfrak{H}^+,\mathbb{C})\times\mathcal{C}_0^{\infty}(\scri^+,\mathbb{C})\cr
\widetilde{{\phi}}_A|_{\Sigma_0}  &\longmapsto (\widetilde{\phi}_0|_{\mathfrak{H}^+}, \widetilde{{\phi}}_1|_{\scri^+}).
\end{align*}
\end{defn}
Using again the energy equality we can extend the domain of the trace operator $\mathcal{T}^+$, where the extended operator is one-to-one and has closed range.
\begin{cor}\label{trace-one-one} 
We extend the trace operator 
\begin{align*}
\mathcal{T}^+: \mathcal{H}_0 = L^2(\Sigma_0, \mathbb{S}_A) &\longrightarrow  \mathcal{H}^+ = L^2(\mathfrak{H}^+,\mathbb{C}) \times L^2(\scri^+,\mathbb{C})\cr
\widetilde{{\phi}}_A|_{\Sigma_0} &\longmapsto (\widetilde{{\phi}}_0|_{\mathfrak{H}^+},\widetilde{{\phi}}_1|_{\scri^+})
\end{align*}
where $\mathcal{H}_0 = L^2(\Sigma_0,\mathbb{S}_A)$ is the closed space of $\mathcal{C}_0^\infty(\Sigma_0,\mathbb{S}_A)$ in the energy norm 
$$\left\|\widetilde{{\phi}}_A\right\|_{\Sigma_0} = \left( \frac{1}{\sqrt 2} \int_{\Sigma_0} (|\widetilde{\phi}_1|^2  + R^2|\widetilde{\phi}_0|^2)  \d r \d^2\omega \right)^{1/2},$$
and similarly $\mathcal{H}^+ = L^2(\mathfrak{H}^+,\mathbb{C}) \times L^2(\scri^+,\mathbb{C})$ is the closed space of $\mathcal{C}_0^\infty(\mathfrak{H}^+,\mathbb{C}) \times \mathcal{C}_0^\infty(\scri^+,\mathbb{C})$ in the energy norm 
$$\left\|((\sqrt{\Delta}\widetilde{\phi}_0)|_{\mathfrak{H}^+},\widetilde{{\phi}}_1|_{\scri^+})\right\|_{\mathcal{H}^+} = \left( \frac{1}{\sqrt 2}\left(\int_{\mathfrak{H}^+} |(\sqrt{\Delta}\widetilde{\phi}_0)|_{\mathfrak{H}^+}|^2 \d {t^*} \d^2\omega + \int_{\scri^+} |\widetilde{\phi}_1|_{\scri^+}|^2 \d {^*t} \d^2\omega\right)\right)^{1/2}.$$
The trace operator in the new domains is one to one and has closed range.
\end{cor}
\begin{proof}
It is clear that $\mathcal{T}^+$ is one-to-one from the equality energy. Since the equality energy, we have $\mathcal{T}^+$ transforms a Cauchy sequence to another one. Hence, the domain image $\mathcal{T}^+(L^2(\Sigma_0,\mathbb{S}_A))$ is closed.
\end{proof}

\section{Goursat problem and conformal scattering operator}\label{S5}
\subsection{The full field on conformal boundaries}
We consider the Gousat problem in the future $\mathcal{I}^+(\Sigma_0) \subset \bar{\mathcal{B}}_I$:
\begin{align}\label{Goursat1}
\begin{cases}
\widetilde{\nabla}^{AA'} \widetilde{\phi}_A &= 0,\cr
\widetilde{\phi}_1|_{\scri^+} &= \widetilde{\psi}_1 \in \mathcal{C}_0^\infty(\scri^+, {\mathbb C}),\,\widetilde{\phi}_A|_{\scri^+} = \widetilde{\psi}_A \in \mathcal{D}_{\scri^+},\cr
\widetilde{\phi}_0|_{\mathfrak{H}^+} &= \widetilde{\psi}_0 \in \mathcal{C}_0^\infty(\scri^+, {\mathbb C}),\,\widetilde{\phi}_A|_{\mathfrak{H}^+} = \widetilde{\bar\psi}_A \in \mathcal{D}_{\mathfrak{H}^+},
\end{cases}
\end{align}
here $\mathcal{D}_{\scri^+}$ and $\mathcal{D}_{\mathfrak{H}^+}$ are the constraint spaces on $\scri^+$ and $\mathfrak{H}^+$ respectively.

We recall the first equation of the sytem \eqref{Scalar} of the massless equation $\widetilde{\nabla}^{AA'} \widetilde{\phi}_A = 0$ on the rescaled star-Kerr coordinates $({^*t},R,\theta,{^*\varphi})$:
\begin{gather*}
\sqrt{\frac{2}{\Delta\rho^2}}\left( (r^2+a^2)\partial_{^*t} + a\partial_{^*\varphi} + \frac{R^2\Delta}{2}\partial_R \right)\widetilde{\phi}_0 - \frac{r}{\sqrt{2}p}\left( ia\sin\theta\partial_{^*t}+\partial_{\theta}+\frac{i}{\sin\theta}\partial_{^*\varphi} \right)\widetilde{\phi}_1 \cr
+\left( \left( 2R - \frac{r}{\rho^2} \right)\sqrt{\frac{\Delta}{2\rho^2}} - \frac{Mr^2-a^2(r\sin^2\theta+M\cos^2\theta)}{2\rho^2\sqrt{2\Delta\rho^2}} \right)\widetilde{\phi}_0 \cr
-\frac{r}{\sqrt{2}p}\left( \frac{ia\sin\theta}{\bar{p}} + \frac{\cot\theta}{2} + \frac{a^2\sin\theta\cos\theta}{2\rho^2}\right)\widetilde{\phi}_1 = 0.
\end{gather*}
Since the constraint system on $\scri^+$ is the projection of the equation $\widetilde{\nabla}^{AA'}\widetilde{{\phi}}_A = 0$ on the null normal vector $\widetilde{n}^a$, the constraint on the null infinity hypersurface $\scri^+$ is that of the above equation on $\scri^+$:
\begin{equation*}
\sqrt 2 \partial_{^*t} \widetilde{\phi}_0|_{\scri^+} - \frac{1}{\sqrt 2}\left( ia\sin \theta \partial_{^*t} + \partial_\theta + \frac{i}{\sin\theta}\partial_\varphi + \frac{1}{2}\cos\theta \right) \widetilde{\phi}_1|_{\scri^+} = 0.
\end{equation*}
Therefore on $\scri^+$, we have
\begin{equation*}
\widetilde{\phi}_0|_{\scri^+}(^*t) = \widetilde{\phi}_0|_{\scri^+}(-\infty) + \frac{1}{2}\int^{^*t}_{-\infty}\left( ia\sin \theta \partial_{^*t}+ \partial_\theta + \frac{i}{\sin\theta}\partial_{^*\varphi}  + \frac{1}{2}\cos\theta  \right)\widetilde{\phi}_1|_{\scri^+}(s) \d s.
\end{equation*}
Therefore, we have to obtain the fully spin field $\widetilde{\phi}_A|_{\scri^+}$ on the future null infinity $\scri^+$.
Since the initial data $\widetilde{\psi}_1 \in \mathcal{C}^\infty_0({\scri^+}, \mathbb C)$ has the support away from $i^+$ and $i^0$, we have that the support of $\widetilde{\phi}_0$ is also far away from $i^+$ and $i_0$.

Similarly, we can obtain the fully spin field $\widetilde{\phi}_A|_{\mathfrak{H}^+}$ with the support is far away from $i^+$ from $\widehat{\psi}_0 = R\widetilde{\psi}_0 \in \mathcal{C}_0^\infty(\scri^+, {\mathbb C})$. 
In the rescaled Kerr-star coordinates $(t^*,R,\theta,\varphi^*)$ we consider the second equation of \eqref{Scalar1}:
\begin{gather*}
\sqrt{\frac{2}{\Delta\rho^2}}\left( (r^2+a^2)\partial_{t^*} + a\partial_{\varphi^*} - \frac{R^2\Delta}{2}\partial_R \right)\widehat{\phi}_1 \cr
+ \frac{r}{\sqrt{2}\bar{p}}\left( ia\sin\theta\partial_{t^*} - \partial_{\theta} + \frac{i}{\sin\theta}\partial_{\varphi^*} -\frac{\cot\theta}{2} + \frac{a^2\sin\theta\cos\theta}{2\rho^2}\right)\widehat{\phi}_0 \cr
+ \left( \frac{Mr^4 - a^2r^2(r\sin^2\theta+M\cos^2\theta)}{2\rho^2\sqrt{2\Delta\rho^2}} + \frac{iar\cos\theta}{\bar{p}}\sqrt{\frac{\Delta}{2\rho^2}} \right)\widehat{\phi}_1 = 0.
\end{gather*}
Hence
\begin{gather*}
\sqrt{2}\left( \partial_{t^*} + \frac{a}{r^2+a^2}\partial_{\varphi^*} - \frac{R^2\sqrt{\Delta}}{2(r^2+a^2)\sqrt{\rho^2}}\partial_R \right)\widehat{\phi}_1 \cr
+ \frac{r\sqrt{\Delta\rho^2}}{\sqrt{2}\bar{p}(r^2+a^2)}\left( ia\sin\theta\partial_{t^*} - \partial_{\theta} + \frac{i}{\sin\theta}\partial_{\varphi^*} -\frac{\cot\theta}{2} + \frac{a^2\sin\theta\cos\theta}{2\rho^2}\right)\widehat{\phi}_0 \cr
+ \left( \frac{Mr^4 - a^2r^2(r\sin^2\theta+M\cos^2\theta)}{2\sqrt{2}\rho^2(r^2+a^2)} + \frac{iar\cos\theta}{\bar{p}}\frac{\Delta}{\sqrt{2}(r^2+a^2)} \right)\widehat{\phi}_1 = 0.
\end{gather*}
Taking the constraint of this equation on $\mathfrak{H}^+$ we get
\begin{gather*}
\sqrt{2}\left( \partial_{t^*} + \frac{a}{r_+^2+a^2}\partial_{\varphi^*} \right)\widehat{\psi}_1|_{\mathfrak{H}^+} \cr
+ \frac{r_+\sqrt{r_+^2+a^2\cos^2\theta}}{\sqrt{2}(r_+ + ia\sin\theta)(r_+^2+a^2)}\left( ia\sin\theta\partial_{t^*} - \partial_{\theta} + \frac{i}{\sin\theta}\partial_{\varphi^*} -\frac{\cot\theta}{2} + \frac{a^2\sin\theta\cos\theta}{2(r_+^2+a^2\cos^2\theta)}\right)\lim_{r\to r_+}\sqrt{\Delta}\widehat{\phi}_0 \cr
+ \frac{Mr_+^4 - a^2r_+^2(r_+\sin^2\theta+M\cos^2\theta)}{2\sqrt{2}(r_+^2+a^2\cos^2\theta)(r_+^2+a^2)} \widehat{\psi}_1|_{\mathfrak{H}^+} = 0.
\end{gather*}
This is a transported type equation with the unknown $\widehat{\psi}_1|_{\mathfrak{H}^+}(t^*,\varphi^*)$:
\begin{equation}\label{TranEq}
\partial_{t^*}\widehat{\psi}_1|_{\mathfrak{H}^+} + \frac{a}{r_+^2 + a^2}\partial_{\varphi^*}\widehat{\psi}_1|_{\mathfrak{H}^+} = f\widehat{\psi}_1 + g(\lim_{r\to r_+}\sqrt{\Delta}\widehat{\psi}_0),
\end{equation}
where 
$$f = -\frac{Mr_+^4 - a^2r_+^2(r_+\sin^2\theta+M\cos^2\theta)}{4(r_+^2+a^2\cos^2\theta)(r_+^2+a^2)} = -\frac{r_+^2 - Mr_+}{8M}, $$
$$g = \frac{r_+\sqrt{r_+^2+a^2\cos^2\theta}}{2(r_+ + ia\sin\theta)(r_+^2+a^2)}\left( ia\sin\theta\partial_{t^*} - \partial_{\theta} + \frac{i}{\sin\theta}\partial_{\varphi^*} -\frac{\cot\theta}{2} + \frac{a^2\sin\theta\cos\theta}{2(r_+^2+a^2\cos^2\theta)}\right).$$
Equation \eqref{TranEq} can be solved by the method of characteristics and we get
\begin{equation}\label{Phi1eq}
\widehat{\psi}_1|_{\mathfrak{H}^+}(t^*,\varphi^*) = e^{ft^*} \int g\left(t^*,C,\lim_{r\to r_+} \sqrt{\Delta}\widehat{\psi}_0 \right)e^{-ft^*}\d t^*,
\end{equation}
where the characteristic curve is given by
$$C= \frac{a}{r_+^2 + a^2}t^* -\varphi^*.$$
Combining with the initial data $\widehat{\psi}_1|_{\mathfrak{H}^+}(-\infty):=\widehat{\psi}_1|_{\mathfrak{H}^+}(t^* = -\infty)$ we can find $C$ and obtain the solution $\widehat{\psi}_1|_{\mathfrak{H}^+}$, hence $\widetilde{\psi}_1|_{\mathfrak{H}^+} = R\widehat{\psi}_1|_{\mathfrak{H}^+}$ depend on $\widetilde{\psi}_1|_{\mathfrak{H}^+}(-\infty)$ and $\lim_{r\to r_+} \sqrt{\Delta}R\widetilde{\psi}_0$.
Therefore, we obatin the fully spin field $\widetilde{\phi}_A|_{\mathfrak{H}^+}$ with the support is far away from $i^+$ from $\widetilde{\psi}_0 \in \mathcal{C}_0^\infty(\scri^+, {\mathbb C})$.

\subsection{Solving Goursat problem in the future $\mathcal {I}^+(\mathcal{S})$}
Let ${\mathcal {S}}$ be a spacelike hypersurface in $\widetilde{\mathcal{B}}_I$ such that $\mathcal{S}$ pass the  bifurcation sphere $S^2_C$ and crosses $\scri^+$ strictly in the past of the support of the initial data $\widetilde{\psi}_1$, we denote the point of intersection of $\mathcal{S}$ and $\scri^+$ by $Q$.
\begin{figure}[H]
\begin{center}
\includegraphics[scale=0.7]{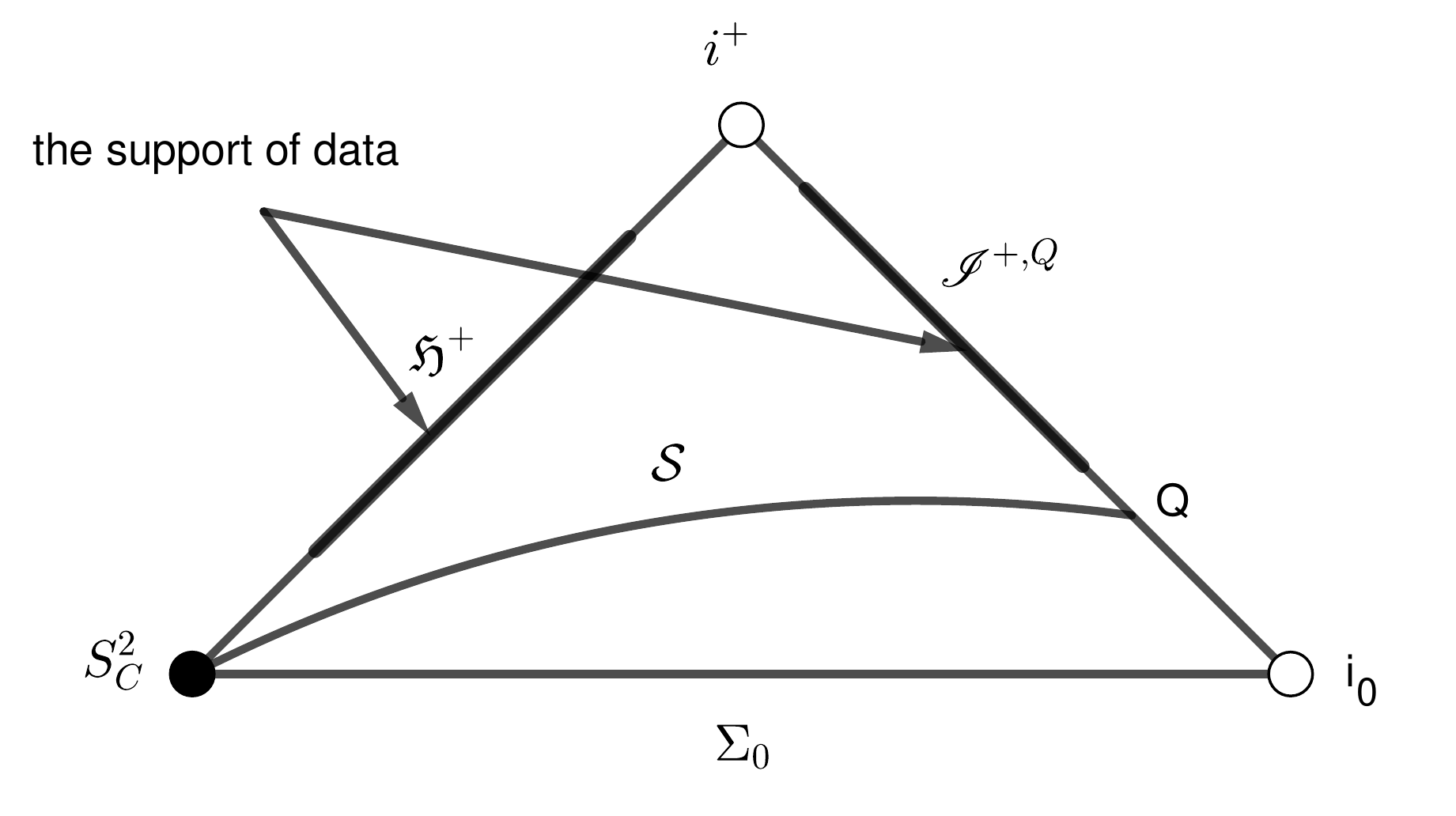}
\caption{The spacelike hypersurface $\mathcal{S}$.}
\end{center}
\end{figure}
Now we solve the Goursat problem in the future $\mathcal {I}^+(\mathcal{S})$ of ${\mathcal S}$. We have (see equation \eqref{app_3_wave0}):
$$2\widetilde{\nabla}_{ZA'}\widetilde{\nabla}^{AA'} \widetilde{\phi}_A = \check{\Box}\widetilde{\phi}_Z  + \widetilde{X}_{ZA}{^{NA}} \widetilde{\phi}_N = 0.$$
Therefore, the Goursat problem on the future $\mathcal{I}^+(\mathcal{S})$ has a problem consequence as follows
\begin{align}\label{Goursat2}
\begin{cases}
\check{\Box} \widetilde{\phi}_Z + \widetilde{X}_{ZA}{^{NA}} \widetilde{\phi}_N &= 0,\cr
\widetilde{\phi}_A|_{\scri^{+,Q}} &= \widetilde{\psi}_A|_{\scri^{+,Q}} \in \mathcal{C}^\infty_0(\scri^{+,Q}, {\mathbb S}_A),\,\widetilde{\nabla}^{AA'}\widetilde{\phi}_A|_{\scri^{+,Q}}  = 0,\cr
\widetilde{\phi}_A|_{\mathfrak{H}^+} &= \widetilde{\bar{\psi}}_A|_{\mathfrak{H}^+} \in \mathcal{C}^\infty_0(\mathfrak{H}^+, {\mathbb S}_A),\,\widetilde{\nabla}^{AA'}\widetilde{\bar{\phi}}_A|_{\mathfrak{H}^+} = 0,
\end{cases}
\end{align}
where the coefficients of the term $\widetilde{X}_{ZA}{^{NA}} \widetilde{\phi}_N$ are $\mathcal{C}^\infty$ by Lemma \ref{CurSpin} and
$\scri^{+,Q}$ is the future part of $Q$ in the null infinity hypersurface $\scri^+$. Applying the generalized result of H\"ormander for the spin wave equation (see Appendix \ref{app_3_Goursat}) the system \eqref{Goursat2} has a unique solution $\widetilde{\phi}_A \in \mathcal{I}^+(\mathcal{S})/\mathcal{V}$, where $\mathcal{V}$ is a neighbourhood of a point lying in the future of the support of the Goursat data.

Now we show that $\widetilde{\phi}_A$ is also a solution of the system \eqref{Goursat1} by proving that $\widetilde{\nabla}^{AA'} \widetilde{\phi}_A = 0$. First, the components of $\widetilde{\nabla}^{AA'} \widetilde{\psi}_A$, i.e, the restrictions of the components of $\widetilde{\nabla}^{AA'} \widetilde{\phi}_A$ on the hypersurface $\scri^{+,Q}$ are both zero. Indeed, if we set
$$\Xi^{A'}:= \widetilde{\nabla}^{AA'} \widetilde{\phi}_A,$$
then we have
$$\Xi^{1'}|_{\scri^{+,Q}} = \widetilde{\iota}_{A'}\Xi^{A'}|_{\scri^{+,Q}} = \widetilde{\nabla}^{AA'} \widetilde{\phi}_A|_{\scri^{+,Q}} = 0.$$
Hence $\Xi^{1'}|_{\scri^{+,Q}}=0$. On the other hand, by the equation 
$$\widetilde{\Box} \widetilde{\phi}_A + \widetilde{X}_{ZA}{^{NA}} \widetilde{\phi}_N = \frac{1}{2} \widetilde{\nabla}_{AK'}\widetilde{\nabla}^{KK'} \widetilde{\phi}_K = \frac{1}{2} \widetilde{\nabla}_{AK'}\Xi^{K'} = \frac{1}{2} \Theta_A =  0.$$
we have
$$\Theta_1 = \Theta_0 = 0,$$
where $\Theta_1$ and $\Theta_0$ are obtained by the differential equations which are of order one in the components of $\Xi^{1'}$ and $\Xi^{0'}$ (for detailed calculations see Appendix \ref{app_3_express}). Taking the constraint of these equations on $\scri^{+,Q}$ we obtain the restrictive equations of $\Xi^{1'}$ and $\Xi^{0'}$ on $\scri^{+,Q}$. Since $\Xi^{1'} =0 $ on $\scri^{+,Q}$, we can obtain the Cauchy problem of the system of differential equations of order one, where the unknowns are only the restriction of $\Xi^{0'}$ on $\scri^{+,Q}$:
\begin{align}
\begin{cases}
\Theta_1|_{\scri^+} &= 0,\cr
\Xi^{0'}|_{{\mathcal V}(P)} &= 0 
\end{cases}
\end{align}
where ${\mathcal V}(P)$ is the neighborhood of the point $P$ chosen to belong to $\scri^{+,Q}$, near $i^+$ and not belonging to the support of $\widetilde{\psi}_A$. Since the Cauchy problem has a unique solution, we have that $\Xi^{0'}|_{\scri^{+,Q}}$ is zero. Therefore we have that the restrictions of the components of $\widetilde{\nabla}^{AA'} \widetilde{\phi}_A$ on $\scri^{+,Q}$ are both zero (see Appendix \ref{app_3_express}).

Now we have (see Equation \eqref{app_3_wave1} in Appendix)
$$0 = \widetilde{\nabla}^{AA'}\check{\Box}\widetilde{\phi}_A = \frac{1}{2} \widetilde{\nabla}^{AA'}\widetilde{\nabla}_{AK'}\Xi^{K'} = \frac{1}{4} \breve{\Box}\Xi^{A'} + \frac{1}{4} \widetilde{{\bar X}}^{A'}{_{K'Q'}}{^{K'}} \Xi^{Q'}$$
and then
\begin{align}\label{1}
\begin{cases}
\breve{\Box}\Xi^{A'} + \widetilde{{\bar X}}^{A'}{_{K'Q'}}{^{K'}} \Xi^{Q'}&= 0,\cr
\mbox{The restrictions of all the components of} \; \Xi^{A'} \; \mbox{on} \; \scri^{+,Q} &= 0 
\end{cases}
\end{align}
where $\widetilde{X}_{ABCD}$ is the curvature spinor
$$\widetilde{X}_{ABCD} = \widetilde{\Psi}_{ABCD} + \widetilde{\Lambda} (\widetilde{\varepsilon}_{AC}\widetilde{\varepsilon}_{BD} + \widetilde{\varepsilon}_{AD}\widetilde{\varepsilon}_{BC})= {\Psi}_{ABCD} + \widetilde{\Lambda} (\widetilde{\varepsilon}_{AC}\widetilde{\varepsilon}_{BD} + \widetilde{\varepsilon}_{AD}\widetilde{\varepsilon}_{BC}).$$
The components of $\Psi_{ABC'D'}$ are $\mathcal {C}^{\infty}$ by Lemma \ref{CurSpin}.

By the same way as above we have also that (see the last of Appendix \ref{app_3_express}):
\begin{align}\label{2}
\begin{cases}
\breve{\Box}\Xi^{A'} + \widetilde{{\bar X}}^{A'}{_{K'Q'}}{^{K'}} \Xi^{Q'}&= 0,\cr
\mbox{The restrictions of all the components of} \; \Xi^{A'} \; \mbox{on} \; \mathfrak{H}^+ &= 0. 
\end{cases}
\end{align}

Therefore, since \eqref{1} and \eqref{2} and by using again the generalized result of H\"ormander (Appendix \ref{app_3_Goursat}) with the zero initial data on the null boundary $\mathfrak{H}^+\cup \scri^+_Q$, we get $\Xi^{A'} = 0$ and then $\Xi^{A'} = \widetilde{\nabla}^{AA'} \widetilde{\phi}_A = 0$. So the solution of the system \eqref{Goursat2} is a solution of the system \eqref{Goursat1}. For convenience, we denote by $\widetilde{\phi}^1_A$ the solution of this step.

\subsection{Solving Goursat problem in $\mathcal{I}^-(\mathcal{S})$ and conformal scattering operator}
We need to extend the solution obtained in the previous section down to $\Sigma_0$. This is equivalent to solve the Cauchy problem in the past $\mathcal{I}^-(\mathcal{S})$ of $\mathcal{S}$: 
\begin{align}\label{Cauchy2}
\begin{cases}
\widetilde{\nabla}^{AA'} \widetilde{\phi}_A &= 0,\cr
\widetilde{\phi}_A|_{\mathcal{S}} &= \widetilde{\phi}^1_A|_{\mathcal{S}}.
\end{cases}
\end{align}

As a consequence of Theorem \ref{cauchyproblem}, this Cauchy problem is well-posed, we denote its solution by $\widetilde{\phi}^2_A$ and the solution of this step by $\widetilde{\phi}^2_A$. Clearly, we can obtain by using the Stokes formula \eqref{Stokesformula} and the conservation law $\widetilde{\nabla}^a\widetilde{J}_a =0$ that
$$\mathcal{E}_{\mathcal{S}}(\widetilde{{\phi}}_A) = \mathcal{E}_{\Sigma_0}(\widetilde{\phi}_A).$$
Using the energy equality \eqref{energyequality}, we obtain that
$$\mathcal{E}_{\mathfrak{H}^+}(\widetilde{\phi}_A) + \mathcal{E}_{\scri^{+,Q}}(\widetilde{\phi}_A)= \mathcal{E}_{\mathcal{S}}(\widetilde{{\phi}}_A) = \mathcal{E}_{\Sigma_0}(\widetilde{\phi}_A).$$
Therefore the energy of the solution on the hypersurface $\Sigma_0$ is finite and we can define the trace operator as the constraint of the solution of the Cauchy problem \eqref{Cauchy2} on $\Sigma_0$.

Finally, the solution of the Goursat problem is the union of the solutions
$$\widetilde{\phi}_A= \begin{cases}\widetilde{\phi}^1_A \mbox{ in the domain $\mathcal{I}^+(S)$},\cr 
\widetilde{\phi}^2_A\mbox{ in the domain $\mathcal{I}^-(S)$}.
\end{cases}$$
By the solving of Goursat problem $\mathcal{I}^+(\mathcal{S})$ and $\mathcal{I}^-(\mathcal{S})$ we have obtained the following theorem
\begin{theorem}(Goursat problem)\label{Goursatprob}
The Goursat problem for the rescaled spin-$1/2$ massless equation $\widetilde{\nabla}^{AA'}\widetilde{\phi}_A = 0$ in $\bar{\mathcal B}_I$ is well-posed, i.e, for any $(\widetilde{\psi}_1, \widetilde{\psi}_0) \in \mathcal{C}_0^\infty(\scri^+) \times \mathcal{C}_0^\infty(\mathfrak{H}^+)$, $\widetilde{\phi}_A|_{\scri^+} = \widetilde{\psi}_A \in \mathcal{D}_{\scri^+}$ and $\widetilde{\phi}_A|_{\mathfrak{H}^+} = \widetilde{\bar{\psi}}_A  \in \mathcal{D}_{\mathfrak{H}^+}$ there exists a unique solution $\widetilde{\phi}_A$ of $\widetilde{\nabla}^{AA'}\widetilde{\phi}_A = 0$ such that   
$$\widetilde{\phi}_A \in {\mathcal C}^\infty(\bar{\mathcal B}_I,\mathbb{S}_A) \; ; \; (\widetilde{\phi}_1,\widetilde{\phi}_0)|_{\scri^+} = (\widetilde{\psi}_1,\widetilde{\psi}_0)$$
and 
$$\widetilde{\phi}_A|_{\scri^+} = \widetilde{\psi}_A,\, \widetilde{\phi}_A|_{\mathfrak{H}^+} = \widetilde{\bar\psi}_A.$$
Furthermore, the energy norm of the constraint of the solution $\widetilde{\phi}_A|_{\Sigma_0}$ on $\Sigma_0$ is finite.
\end{theorem}
We now define the conformal scattering operator for the spin-$1/2$ massless Dirac equation as follows.
\begin{definition}
Similarly, we introduce the past trace operator $\mathcal{T}^-$ and the space $\mathcal{H}^-$ of past scattering data on the past horizon and the past null infinity. We define the scattering operator $S$ as the operator that, to the past scattering data associates the future scattering data, i.e.
$$S:= \mathcal{T}^+\circ (\mathcal{T}^-)^{-1}.$$
\end{definition}
 
\section{Appendix A}\label{A}

\subsection{Goursat problem for the spin wave equations on Kerr spacetime}\label{app_3_Goursat}
In this part we extend the results of H\"ormander \cite{Ho1990} for the spin-$1/2$ wave equations. The results of H\"ormander were extended for the scalar wave equation by Nicolas \cite{Ni2006} with the following minor modifications: the $\mathcal{C}^1$-metric, the continuous coefficients of the derivatives of the first order and the terms of order zero have locally $L^\infty$-coefficients. We refer \cite{Jo2020,Mo2019,Mo2021,Ni2016,Xuan2021,Xuan20,Xuan21} for the appllications of the generalized results of H\"ormander to solve the Goursat problem for the massless spin, tensorial equations, linear and semiliear wave equations. Here we will show that how the Goursat problem is valid for the spin wave equations in the future $\mathcal{I}^+(\mathcal{S})$ of $\mathcal{S}$ in $\bar{\mathcal{B}}_I$ (recall that $\mathcal{S}$ is the spacelike hypersurface in $\bar{\mathcal{B}}_I$ such that it pass $\scri^+$ strictly in the past of the support data). 

Let $P$ be a point in $\bar{\mathcal{B}}_I$, we cut-off $\mathcal{I}^+(\mathcal{S})$ by the future neighbourhood $\mathcal{V}$ of $P$ such that $\mathcal{V}$ does not intersect with the support of Goursat data and get $\mathfrak{B}=\mathcal{I}^+(\mathcal{S})/\mathcal{V}$. We extend $(\mathfrak{B},\widetilde{g})$ onto a cylindrical globally hyperbolic spacetime $(\mathfrak{M}=\mathbb{R}_t\times S^3, \mathfrak{g})$, where $\mathfrak{g}|_{\mathfrak{B}} = \tilde{g}|_{\bar{\mathcal{B}}_I}$ and the part of null conformal boundary $\mathfrak{H}^+\cup \scri^+$ inside $\mathcal{I}^+(\mathcal{S})/\mathcal{V}$ is extended as a null hypersurface $\mathcal{C}$ that is the graph of a Lipschitz function over $S^3$ and the data by zero on the rest of the extended hypersurface.

We consider the Goursat problem of following the spin wave equation in the spacetime $(\mathfrak{M} = \mathbb{R}_t\times S^3, \mathfrak{g})$:
\begin{align}\label{spinwaveeq}
\begin{cases}
\widehat{\Box}\widetilde{\phi}_Z &= 0,\cr
\widetilde{\phi}_A|_{\mathcal{C}} &= \widetilde{\psi}_A|_{\mathcal{C}} \in \mathcal{C}^\infty_0(\mathcal{C}, {\mathbb S}_A),\cr
\nabla_{\mathfrak{g}}^{AA'}\widetilde{\phi}_A|_{\mathcal{C}} & = \widetilde{\zeta}^{A'}|_{\mathcal{C}} \in \mathcal{C}^\infty_0(\mathcal{C}, {\mathbb S}^{A'}),
\end{cases}
\end{align}
where the operator $\widehat{\Box}$ defined by \eqref{WaveOperator1} acts on the full spin fields. Notice that we can replace $\widehat{\Box}$ by the two other spin wave operators $\check{\Box}$ and $\breve{\Box}$ defined in Equations \eqref{WaveOperator2} and \eqref{WaveOperator3} respectively.

Following \cite{Sti1936} the spacetime $(\mathfrak{M} = \mathbb{R}_t\times S^3, \mathfrak{g})$ is parallelizable, i.e, it admit a continuous global frame in the sense that the tangent space at each point has a basis. Therefore, we can chose a global spin-frame $\left\{o,\iota \right\}$ for $\mathcal{M}$ such that in this spin-frame the Newman-Penrose tetrad is $\mathcal {C}^{\infty}$.
Projecting \eqref{spinwaveeq} on $\left\{o,\iota \right\}$ (see the last of Appendix \ref{app_3_express} for the projection of the covariant derivative equation $\nabla_{\mathfrak{g}}^{AA'}\phi_A|_{\mathcal{C}}$) we get the scalar matrix form as follows
\begin{align}\label{system-wave}
\begin{cases}
P\widetilde{\Phi} + L_1 \widetilde{\Phi} &= 0,\cr
(\widetilde{\Phi}, \,\partial_t\widetilde{\Phi})|_{t=0} &= (\widetilde{\Psi}, \, \partial_t\widetilde{\Psi}) \in \mathcal{C}^\infty_0(\mathcal{C})\times \mathcal{C}^\infty_0(\mathcal{C}), 
\end{cases}
\end{align}
where 
$$P= \left(\begin{matrix}
\Box&&0\\
0&&\Box
\end{matrix} \right)$$
is the $2\times 2$-matrix diagram,
$$\widetilde{\Phi} = \left( \begin{matrix} \widetilde{\phi}_0\\ \widetilde{\phi}_1 \end{matrix} \right), \, \widetilde{\Psi} = \left( \begin{matrix} \widetilde{\psi}_0\\ \widetilde{\psi}_1 \end{matrix} \right)$$
is the components of $\widetilde{\phi}_{A}$ and $\widetilde{\Psi}^{A'}$ respectively on the spin-frames $\left\{o_A,\iota_A \right\}$ and $\left\{o^{A'},\iota^{A'} \right\}$ respectively and 
$$L_1 = \left( \begin{matrix} L_1^{00}&&L_1^{01}\\ L_1^{10}&& L_1^{11} \end{matrix}  \right)$$
is the $2\times 2$-matrix where the components are the operators that have the coefficients $\mathcal {C}^\infty$:
$$L_1^{ij} = b_0^{ij}\partial_t + b_\alpha^{ij}\partial_\alpha + c^{ij}.$$

Since $\mathfrak{g}$ is a $\mathcal{C}^1$-metric, the first order terms in $L_1$ have continuous coefficients and the terms of order $0$ have locally $L^\infty$-coefficients, the Goursat problem for the $2\times 2$-matrix wave equation \eqref{system-wave} is well-posed in $(\mathfrak{M}=\mathbb{R}_t\times S^3,\mathfrak{g})$ by applying the results in \cite[Theorem 3 and Theorem 4]{Ni2006}.
\begin{thm}
For the initial data $(\widetilde{\psi}_i,\partial_t\widetilde{\psi}_i) \in {\mathcal C}_0^{\infty}(\mathcal{C})\times {\mathcal{C}}_0^\infty(\mathcal{C}) $ for all $i=0,1$, the $2\times 2$-matrix equation \eqref{system-wave}, hence the spin wave equation \eqref{spinwaveeq} has a unique solution $\widetilde{\Phi} = (\widetilde{\phi}_0,\widetilde{\phi}_1)$ satisfying
$$\widetilde{\phi}_i \in {\mathcal {C}}(\mathbb{R};H^1(S^3))\cap \mathcal{C}^1(\mathbb{R};L^2(S^3)) \; \mbox{for all} \; i = 0,1.$$
\end{thm}
Using the finite propagation speed the solution $\widetilde{\Phi}$ vanishes in $\mathcal{I}^+(\mathcal{S})/\mathfrak{B}$ by local uniqueness and causality. Therefore, the Goursat problem has a unique smooth solution in the future of $\mathcal{S}$, that is the restriction of $\widetilde{\Phi}$ to $\mathfrak{B}$.

\subsection{A pointwise decay for Dirac field on a very slow Kerr spacetime}\label{verySlowKer}
There are some methods to prove the pointwise decay of the field equations such as scalar wave, Dirac, Maxwell, linearized gravity and spin Teukolsky equations on Kerr spacetime (see \cite{AB,ABM,Da2016,Da2019,Ma2020,MeTaTo2012,MeTaTo2017,Ro2020}). The methods base on the transformation equations to the wave (or spin wave) equations with potentials which decay sufficient to establish the decay of solutions.
In this section we develop and apply the results obtained by Tataru et al. in \cite{MeTaTo2012,MeTaTo2017} to establish a pointwise decay for the components of Dirac field in the case $0<a\ll M$.

First, by using the curvature spinors and spinor form of commutators as in Subsections \ref{app_1} and \ref{app_3_Commutator} we can estabish that the origin Dirac field $\phi_A$ satisfies the following spin wave equation 
\begin{equation}\label{OriginSpinWave}
2\nabla_{ZA'}\nabla^{AA'}\phi_A = \check{\Box} \phi_Z + X_{ZA}{^{NA}}\phi_N = 0,
\end{equation}
where $\check{\Box} \phi_Z = \varepsilon^{AM}\nabla_{A'[Z}\nabla{_{M]}}{^{A'}}\phi_A$ and $X_{ZA}{^{NA}}\phi_N = \varepsilon^{AM}\nabla_{A'(Z}\nabla{_{M)}}{^{A'}}\phi_A$.

We notice that the components of the curvature spinor $X_{ZA}{^{NA}}$ can be calculated from the components of Riemann curvature ${R_{ab}}^{cd}$ given in Appendix \ref{Riemann}. We can verify that the components of ${R_{ab}}^{cd}$ are equivalent to $r^{-3}$. Therefore, the coefficients of $X_{ZA}{^{NA}}\phi_N$ are also equivalent to $r^{-3}$. On the other hand, if we use the formulas \eqref{SpinDeri1} and \eqref{SpinDeri2}, then we can transform the spin wave operator $\check{\Box} \phi_Z$ (where $\check{\Box}$ acts on the full spinor field) to the scalar wave operator $\Box_g\phi_Z$ (where $\Box_g$ acts on the functional coefficients) and get
$$\check{\Box} \phi_Z = \Box_g\phi_Z + \left< r\right>^{-2}\partial_\alpha \phi_Z + \left< r\right>^{-3}\phi_Z,$$
where we denote that the functions equivalent to $r$ by $\left<r\right>$. The term $\left< r\right>^{-2}\partial_\alpha \phi_Z + \left< r\right>^{-3}\phi_Z$ arises from the asymptotical part of Kerr metric \eqref{Metric} to Minkowski metric.

Therefore, the spin wave equation \eqref{OriginSpinWave} can be re-written by the following equivalence form
\begin{equation*}
\Box_g\phi_Z = \left< r\right>^{-2}\partial_\alpha \phi_Z + \left< r\right>^{-3}\phi_Z.
\end{equation*}
Projecting the above equation on the global spin-frame $\left\{o_A,\iota_A \right\}$ on Block $\mathcal{B}_I$ we get the scalar matrix form as follows
\begin{equation}\label{SpinCorrEq}
P\Phi = \left< r\right>^{-2}\partial_\alpha \Phi + \left< r\right>^{-3}\Phi, 
\end{equation}
where 
$$P= \left(\begin{matrix}
\Box&&0\\
0&&\Box
\end{matrix} \right)$$
is the $2\times 2$-matrix scalar wave operator and $\Phi = \left( \begin{matrix} \phi_0\\ \phi_1 \end{matrix} \right)$
is the components of $\phi_{A}$ on the spin-frames $\left\{o_A,\iota_A \right\}$. 
\begin{remark}
Under some assumptions on uniform energy bounds and a weak form of local energy decay hold forward in time the pointwise decay of the solution satisfied Equation \eqref{SpinCorrEq} was studied in \cite{MeTaTo2012}. The decays of the coefficients in $L^1$ operator (which consists the derivative terms in order less than or equal to one) arise from the asymptotic part (which contains the angular momentum $a$ and the mass $M$) of Kerr metric \eqref{Metric} to Minkowski metric and they are sufficient to establish the pointwise decay of the solution of \eqref{SpinCorrEq}.
\end{remark}
Equation \eqref{SpinCorrEq} is equivalent to
\begin{equation}\label{phi1}
\Box_g\phi_1 = \left< r\right>^{-2}\partial_\alpha \phi_1 + \left< r\right>^{-3}\phi_1
\end{equation}
and 
\begin{equation}\label{phi00}
\Box_g\phi_0 = \left< r\right>^{-2}\partial_\alpha \phi_0 + \left< r\right>^{-3}\phi_0.
\end{equation}
Now we applying \cite[Theorem 1.5]{MeTaTo2012} for equations \eqref{phi1} and \eqref{phi00} to get the following pointwise decay: suppose $(\phi_1,\partial_t\phi_1)$ is supported inside the cone $C = \left\{t \geq r - R_1 \right\}$ for some $R_1 > 0$, then the following estimates hold
\begin{equation}\label{pointwise1}
|\phi_1|, |\phi_0| \lesssim \frac{1}{\left<t\right>\left< t - r\right>^2}.
\end{equation}
Notice that $\phi_0=\phi_Ao^A$, the constraint of $\nabla^{AA'}\phi_A=0$ on the level set of $t^*$ (i.e, the projection of equation on outgoing null vecto $n^a$) contains $\partial_{^*t}\phi_0$ and $^*t\simeq u:= t-r$ for $r$ large enough. Using these facts we can improve the pointwise decay of $\phi_0$ by the same method for the midle component of Maxwell field in \cite[Section $7.5$]{MeTaTo2017} and get
\begin{equation}\label{ImpPointwise}
|\phi_0| \lesssim \frac{1}{\left< t\right>^2\left< t-r\right>}.
\end{equation}
Indeed, commuting the vector field $r\partial_u:= r(\partial_t - \partial_r)$ to equation\eqref{phi00} we get the following form (the proof is similarly \cite[Lemma $7.2$]{MeTaTo2017}):
$$\Box_g (r\partial_u\phi_0) = \left< r\right>^{-1}\partial_\alpha (\phi_0) + \left< r\right>^{-2}(\phi_0).$$
Therefore, by the same way as in \cite[equation $(7.33)$, page $90$]{MeTaTo2017} we obtain in the region $\left\{ t/2<r<t \right\}$ that 
$$|r\partial_u\phi_0| \lesssim \frac{1}{\left<t\right>\left< t - r\right>^2}.$$
Hence
\begin{equation}\label{3}
|\partial_u\phi_0| \lesssim \frac{1}{\left< r\right>\left<t\right>\left< t - r\right>^2}.
\end{equation}
For the domain $\left\{r<t/2 \right\}$, by using derivative estimates \cite[Theorem $1.5$]{MeTaTo2012} we have 
$$|\partial_u\phi_0| \lesssim \frac{1}{\left< r\right>\left< t - r\right>^3} \lesssim \frac{1}{\left< r\right>\left<t\right>\left< t - r\right>^2}.$$
Therefore, inequality \eqref{3} holds for $r<t$. Integrating \eqref{3} follows $u$ on the level set of $v=t+r$, we get \eqref{ImpPointwise}. 

Finally, the estimates in \eqref{pointwise1} and \eqref{ImpPointwise} satisfy the decay assumption \eqref{pointwise}. Therefore, we can use the results in the previous sections: Theorem \ref{egalite_energies} and Theorem \ref{Goursatprob}  to obtain the conformal scattering theory for the Dirac field in a very slowly Kerr spacetime.

\section{Appendix B}\label{B}

\subsection{Detailed calculations for the Goursat problem}\label{app_3_express}
We have the expression of the spinor field $\Xi^{A'}$ on the spin-frame $\left\{\widetilde{o},\,\widetilde{\iota} \right\}$ as follows
\begin{eqnarray*}
\Xi^{A'} &=& \Xi^{1'} \widetilde{o}^{A'} - \Xi^{0'} \widetilde{\iota}^{A'}.
\end{eqnarray*}
The covariant derivative $\widetilde{\nabla}_{ZA'}$ acts on the full spinor field can be decomposed as
$$\widetilde{\nabla}_a \Xi = (\widetilde{D}\Xi)\widetilde{n}_a + (\widetilde{D}'\Xi)\widetilde{l}_a - (\widetilde{\delta}\Xi)\bar{{\widetilde m}}_a - (\widetilde{\delta}'\Xi)\widetilde{m}_a.$$
We recall that the twelve values of the rescaled spin coefficients are
\begin{equation*} \label{ResSpin-coffi1}
 \widetilde{\kappa}=\widetilde{\sigma}=\widetilde{\lambda}=\widetilde{\nu}=0, 
\end{equation*} 
\begin{equation*} \label{ResSpin-coffi2}
\widetilde{\tau}=-\frac{ia\sin\theta r}{\sqrt 2 \rho^2} , \, \widetilde{\pi}=\frac{ia\sin\theta r}{\sqrt 2 \bar p^2}, \, 
\widetilde{\rho} = -\frac{iar\cos\theta}{\bar p}\sqrt{\frac{\Delta}{2\rho^2}}, \,  \widetilde{\mu} = \left(R-\frac{1}{\bar p}\right)\sqrt{\frac{\Delta}{2\rho^2}},
\end{equation*}
\begin{equation*} \label{ResSpin-coffi3}
\widetilde{\varepsilon} = \frac{Mr^4 - a^2r^2(r\sin^2\theta + M\cos^2\theta)}{2\rho^2\sqrt{2\Delta\rho^2}},\, 
\widetilde{\alpha} = \frac{r}{\sqrt 2\bar{p}}\left(\frac{ia\sin\theta}{\bar p}-\frac{\cot\theta}{2}+\frac{a^2\sin\theta\cos\theta}{2\rho^2}\right),
\end{equation*} 
\begin{equation*} \label{ResSpin-coffi4}
\widetilde{\beta}=\frac {r}{\sqrt 2 p}\left(\frac{\cot\theta}{2}+\frac{a^2\sin\theta\cos\theta}{2\rho^2}\right),
\end{equation*} 
\begin{equation*} \label{ResSpin-coffi5}
\widetilde{\gamma}=\frac{Mr^2 - a^2(r\sin^2\theta + M\cos^2\theta)}{2\rho^2\sqrt{2\Delta \rho^2}} - \left(\frac{ia\cos\theta}{\rho^2}+R\right)\sqrt{\frac{\Delta}{2\rho^2}}.
\end{equation*} 
The covariant derivatie acts on the spin-frame $\left\{ \widetilde{o}_A,\, \widetilde{\iota}_A \right\}$ as (see Equation (4.5.26) in \cite[Vol. 1]{PeRi}):
\begin{eqnarray*}
&&\widetilde{D}\widetilde{o}_A = \widetilde{\varepsilon} \widetilde{o}_A - \widetilde{\kappa}\widetilde{\iota}_A = \frac{Mr^4 - a^2r^2(r\sin^2\theta + M\cos^2\theta)}{2\rho^2\sqrt{2\Delta\rho^2}}\widetilde{o}_A ,\cr
&&\widetilde{D}\widetilde{\iota}_A = -\widetilde{\varepsilon}\widetilde{\iota}_A + \widetilde{\pi}\widetilde{o}_A = - \frac{Mr^4 - a^2r^2(r\sin^2\theta + M\cos^2\theta)}{2\rho^2\sqrt{2\Delta\rho^2}}\widetilde{\iota}_A + \frac{ia\sin\theta r}{\sqrt 2 \bar p^2} \widetilde{o}_A,\cr
&&\widetilde{\delta}'\widetilde{o}_A = \widetilde{\alpha} \widetilde{o}_A - \widetilde{\rho}\widetilde{\iota}_A   =  \frac{r}{\sqrt 2\bar{p}}\left(\frac{ia\sin\theta}{\bar p}-\frac{\cot\theta}{2}+\frac{a^2\sin\theta\cos\theta}{2\rho^2}\right)\widetilde{o}_A + \frac{iar\cos\theta}{\bar p}\sqrt{\frac{\Delta}{2\rho^2}}\widetilde{\iota}_A,\cr
&&\widetilde{\delta}'\widetilde{\iota}_A = -\widetilde{\alpha}\widetilde{\iota}_A + \widetilde{\lambda}\widetilde{o}_A  = -\frac{r}{\sqrt 2\bar{p}}\left(\frac{ia\sin\theta}{\bar p}-\frac{\cot\theta}{2}+\frac{a^2\sin\theta\cos\theta}{2\rho^2}\right) \widetilde{\iota}_A,\cr
&&\widetilde{\delta}\widetilde{o}_A = \widetilde{\beta} \widetilde{o}_A - \widetilde{\sigma}\widetilde{\iota}_A   = \frac {r}{\sqrt 2 p}\left(\frac{\cot\theta}{2}+\frac{a^2\sin\theta\cos\theta}{2\rho^2}\right) \widetilde{o}_A,\cr
&&\widetilde{\delta}\widetilde{\iota}_A = -\widetilde{\beta}\widetilde{\iota}_A + \widetilde{\mu}\tilde{o}_A = -\frac {r}{\sqrt 2 p}\left(\frac{\cot\theta}{2}+\frac{a^2\sin\theta\cos\theta}{2\rho^2}\right) \widetilde{\iota}_A + \left(R-\frac{1}{\bar p}\right)\sqrt{\frac{\Delta}{2\rho^2}}\widetilde{\iota}_A,\cr
&&\widetilde{D}'\widetilde{o}_A = \widetilde{\gamma} \widetilde{o}_A - \widetilde{\tau}\widetilde{\iota}_A = \left(\frac{Mr^2 - a^2(r\sin^2\theta + M\cos^2\theta)}{2\rho^2\sqrt{2\Delta \rho^2}} - \left(\frac{ia\cos\theta}{\rho^2}+R\right)\sqrt{\frac{\Delta}{2\rho^2}}\right)\widetilde{o}_A -\frac{ia\sin\theta r}{\sqrt 2 \rho^2}\widetilde{\iota}_A, \cr
&&\widetilde{D}'\widetilde{\iota}_A = -\widetilde{\gamma}\widetilde{\iota}_A + \widetilde{\nu}\widetilde{o}_A  = -\left(\frac{Mr^2 - a^2(r\sin^2\theta + M\cos^2\theta)}{2\rho^2\sqrt{2\Delta \rho^2}} - \left(\frac{ia\cos\theta}{\rho^2}+R\right)\sqrt{\frac{\Delta}{2\rho^2}}\right) \widetilde{\iota}_A.
\end{eqnarray*}
Similarly on the dual conjugation spin-frame $\left\{ \widetilde{o}^{A'}, \, \widetilde{\iota}^{A'} \right\}$ we have
\begin{eqnarray*}
&&\widetilde{D}\widetilde{o}^{A'} = \frac{Mr^4 - a^2r^2(r\sin^2\theta + M\cos^2\theta)}{2\rho^2\sqrt{2\Delta\rho^2}}\widetilde{o}^{A'} ,\cr
&&\widetilde{D}\widetilde{\iota}^{A'} = - \frac{Mr^4 - a^2r^2(r\sin^2\theta + M\cos^2\theta)}{2\rho^2\sqrt{2\Delta\rho^2}}\widetilde{\iota}^{A'} + \frac{ia\sin\theta r}{\sqrt 2 \bar p^2} \widetilde{o}^{A'},\cr
&&\widetilde{\delta}'\widetilde{o}^{A'}   =  \frac{r}{\sqrt 2\bar{p}}\left(\frac{ia\sin\theta}{\bar p}-\frac{\cot\theta}{2}+\frac{a^2\sin\theta\cos\theta}{2\rho^2}\right)\widetilde{o}^{A'} + \frac{iar\cos\theta}{\bar p}\sqrt{\frac{\Delta}{2\rho^2}}\widetilde{\iota}^{A'},\cr
&&\widetilde{\delta}'\widetilde{\iota}^{A'} = -\frac{r}{\sqrt 2\bar{p}}\left(\frac{ia\sin\theta}{\bar p}-\frac{\cot\theta}{2}+\frac{a^2\sin\theta\cos\theta}{2\rho^2}\right) \widetilde{\iota}^{A'},\cr
&&\widetilde{\delta}\widetilde{o}^{A'} = \frac {r}{\sqrt 2 p}\left(\frac{\cot\theta}{2}+\frac{a^2\sin\theta\cos\theta}{2\rho^2}\right) \widetilde{o}^{A'},\cr
&&\widetilde{\delta}\widetilde{\iota}^{A'} = -\frac {r}{\sqrt 2 p}\left(\frac{\cot\theta}{2}+\frac{a^2\sin\theta\cos\theta}{2\rho^2}\right) \widetilde{\iota}^{A'} + \left(R-\frac{1}{\bar p}\right)\sqrt{\frac{\Delta}{2\rho^2}}\widetilde{\iota}^{A'},\cr
&&\widetilde{D}'\widetilde{o}^{A'} = \left(\frac{Mr^2 - a^2(r\sin^2\theta + M\cos^2\theta)}{2\rho^2\sqrt{2\Delta \rho^2}} - \left(\frac{ia\cos\theta}{\rho^2}+R\right)\sqrt{\frac{\Delta}{2\rho^2}}\right)\widetilde{o}^{A'} -\frac{ia\sin\theta r}{\sqrt 2 \rho^2}\widetilde{\iota}^{A'}, \cr
&&\widetilde{D}'\widetilde{\iota}^{A'} = -\left(\frac{Mr^2 - a^2(r\sin^2\theta + M\cos^2\theta)}{2\rho^2\sqrt{2\Delta \rho^2}} - \left(\frac{ia\cos\theta}{\rho^2}+R\right)\sqrt{\frac{\Delta}{2\rho^2}}\right) \widetilde{\iota}^{A'}.
\end{eqnarray*}
We have the detailed expression of $\widetilde{\nabla}_{ZA'}\Xi^{A'}$ as 
\begin{eqnarray}\label{Re}
&&\widetilde{\nabla}_{ZA'}\Xi^{A'} = (\widetilde{D}\Xi^{A'})\widetilde{n}_a + (\widetilde{D}'\Xi^{A'})\widetilde{l}_a - (\widetilde{\delta}\Xi^{A'})\bar{{\widetilde m}}_a - (\widetilde{\delta}'\Xi^{A'})\widetilde{m}_a \cr
&=& \widetilde{D} \left( \Xi^{1'}\widetilde{o}^{A'} \right) \widetilde{\iota}_A\widetilde{\iota}_{A'} - \widetilde{D}' \left( \Xi^{0'}\widetilde{\iota}^{A'} \right) \widetilde{o}_A \widetilde{o}_{A'} + \widetilde{\delta} \left( \Xi^{0'} \widetilde{\iota}^{A'} \right) \widetilde{\iota}_{A} \widetilde{o}_{A'} - \widetilde{\delta}' \left( \Xi^{1'}\widetilde{o}^{A'} \right) \widetilde{o}_A \widetilde{\iota}_{A'} \cr
&=& \left\{ \left( -\widetilde{D} - \frac{Mr^4 - a^2r^2(r\sin^2\theta + M\cos^2\theta)}{2\rho^2\sqrt{2\Delta\rho^2}} \right) \Xi^{1'} +  \left( \widetilde{\delta} - \frac {r}{\sqrt 2 p}\left(\frac{\cot\theta}{2}+\frac{a^2\sin\theta\cos\theta}{2\rho^2}\right) \right) \Xi^{0'}\right\} \widetilde{\iota}_A \cr
&&+ \left\{ \left( -\widetilde{D}' + \left(\frac{Mr^2 - a^2(r\sin^2\theta + M\cos^2\theta)}{2\rho^2\sqrt{2\Delta \rho^2}} - \left(\frac{ia\cos\theta}{\rho^2}+R\right)\sqrt{\frac{\Delta}{2\rho^2}}\right) \right)\Xi^{0'}\right.\cr
&&\hspace{2cm}\left. + \left( \widetilde{\delta}' + \frac{r}{\sqrt 2\bar{p}}\left(\frac{ia\sin\theta}{\bar p}-\frac{\cot\theta}{2}+\frac{a^2\sin\theta\cos\theta}{2\rho^2}\right) \right)\Xi^{1'}  \right\}  \widetilde{o}_A.
\end{eqnarray}
Taking the constrain of the system \eqref{Re} on $\scri^+$ with noting that $\Xi^{1'}|_{\scri^+} =0$ we get only the constraint of the second equations
$$-\lim_{r\to \infty}\widetilde{D}'\Xi^{0'}|_{\scri^+} = -\lim_{r\to \infty}\sqrt{\frac{2(a^2+r^2)}{\Delta}} \partial_{^*t} \Xi^{0'}|_{\scri^+} = \sqrt{2}\partial_{^*t} \Xi^{0'}|_{\scri^+}= 0.$$
Integrating these equations along $\scri^+$, we get $\Xi^{0'}|_{\scri^+} = \; constant$. This leads to a fact that the Cauchy problem with the initial condition $\Xi^{0'}|_{{\mathcal V}(P)\cap \scri^+} = 0$ has a unique solution and it equals to zero.

By the same way we can obatin the restriction of the rescaled equation on $\mathfrak{H}^+$ by considering the rescaled equation $\widehat{\nabla}_{ZA'}\Xi^{A'}=0$ on the rescaled Kerr-star coordinates $(t^*,R,\theta,\varphi^*)$. By the same calculations we get the constraint on $\mathfrak{H}^+$ of $\widehat{\nabla}_{ZA'}\Xi^{A'}=0$ is the constraint of the first equation of \eqref{Re} on $\mathfrak{H}^+$ with $\widetilde{.}$ replacing by $\widehat{.}$:
$$\left( -\widehat{D} - \frac{Mr^4 - a^2r^2(r\sin^2\theta + M\cos^2\theta)}{2\rho^2\sqrt{2\Delta\rho^2}} \right) \Xi^{1'} +  \left( \widehat{\delta} - \frac {r}{\sqrt 2 p}\left(\frac{\cot\theta}{2}+\frac{a^2\sin\theta\cos\theta}{2\rho^2}\right) \right) \Xi^{0'} = 0.$$
Multiplying the above equation with $\dfrac{\sqrt{\Delta\rho^2}}{\sqrt{2}(r_+^2+a^2)}$ and taking the constraint of the obtained equation on $\mathfrak{H}^+$ with noting that $\Xi^{0'}|_{\mathfrak{H}^+} =0$ we get
$$\left( -\partial_{t^*} - \frac{a}{r_+^2 + a^2}\partial_{\varphi^*} - \frac{Mr_+^4 - a^2r_+^2(r_+\sin^2\theta + M\cos^2\theta)}{4(r_+^2+a^2)(r_+^2+ a^2\cos^2\theta)} \right) \Xi^{1'}|_{\mathfrak{H}^+} = 0.$$
Using $\Delta|_{\mathfrak{H}^+} = r_+^2 - 2Mr_+ + a^2 = 0$ we obtain that
$$\left( -\partial_{t^*} - \frac{a}{2Mr_+}\partial_{\varphi^*} - \frac{r_+^2-Mr_+}{8M} \right) \Xi^{1'}|_{\mathfrak{H}^+} = 0.$$
This is a tranported type equation and it has a trivial solution with the initial condition $\Xi^{1'}|_{\mathcal{V}(P)\cap \mathfrak{H}^+} = 0$. This equation has also non trivial solution given by
$$\Xi^{1'}|_{\mathfrak{H}^+} = e^{h(C)}e^{-\frac{r_+^2 - Mr_+}{8M}t^*}, \,\,\, C= t^* - \frac{a}{2Mr_+}\varphi^*.$$
However, this solution does not satisfy that the initial condition $\Xi^{1'}|_{\mathcal{V}(P)\cap \mathfrak{H}^+} = 0$ (this is only true at $t^* = +\infty$). Therefore, we conclude that if the support of $\Xi^{1'}|_{\mathfrak{H}^+}$ is compact and far away from $i^+$, then the rescaled equation $\widehat{\nabla}_{ZA'}\Xi^{A'}=0$ leads to $\Xi^{1'}|_{\mathfrak{H}^+} = 0$.

\subsection{Spin coefficients and derivations of the origin spin-frame}\label{app_4_spin}
Using the Newman-Penrose tetrad normalization \eqref{New2}, the spin coefficients are calculated (see \cite{HaNi2004}):
\begin{equation*} \label{spin0}
\kappa= \tilde{\sigma} = \lambda = \nu = 0, 
\end{equation*} 
\begin{equation*} \label{spin1}
\tau = -\frac{ia\sin\theta}{\sqrt 2 \rho^2} , \, \pi = \frac{ia\sin\theta}{\sqrt 2 \bar p^2}, \, 
\tilde{\rho} = \mu = -\frac{1}{\bar p}\sqrt{\frac{\Delta}{2\rho^2}}, \, \varepsilon = \frac{Mr^2 - a^2(r\sin^2\theta + M\cos^2\theta)}{2\rho^2\sqrt{2\Delta\rho^2}},
\end{equation*}
\begin{equation*} \label{spin2}
\alpha = \frac{1}{\sqrt{2}\bar{p}}\left( \frac{ia\sin\theta}{\bar {p}}-\frac{\cot\theta}{2} + \frac{a^2\sin\theta\cos\theta}{2\rho^2} \right), \, \beta = \frac{1}{\sqrt{2}p} \left( \frac{\cot\theta}{2} +\frac{a^2\sin\theta\cos\theta}{2\rho^2} \right),
\end{equation*} 
\begin{equation*} \label{spin3}
\gamma=\frac{Mr^2 - a^2(r\sin^2\theta + M\cos^2\theta)}{2\rho^2\sqrt{2\Delta \rho^2}} - \frac{ia\cos\theta}{\rho^2}\sqrt{\frac{\Delta}{2\rho^2}}, 
\end{equation*} 
here we use $\tilde{\sigma}$ and $\tilde{\rho}$ to avoid the confusion with the parameters $\sigma$ and $\rho$ in the Kerr metric.

The covariant derivatie acts on the spin-frame $\left\{o_A,\, \iota_A \right\}$ as (see Equation (4.5.26) in \cite[Vol. 1]{PeRi}):
\begin{equation}\label{SpinDeri1}
\begin{gathered}
Do_A = \varepsilon o_A - \kappa\iota_A = \frac{Mr^2 - a^2(r\sin^2\theta + M\cos^2\theta)}{2\rho^2\sqrt{2\Delta\rho^2}}o_A ,\cr
D\iota_A = -\varepsilon\iota_A + \pi o_A = -\frac{Mr^2 - a^2(r\sin^2\theta + M\cos^2\theta)}{2\rho^2\sqrt{2\Delta\rho^2}}\iota_A + \frac{ia\sin\theta}{\sqrt 2 \bar p^2} o_A,\cr
\delta'o_A = \alpha o_A - \tilde{\rho}\iota_A   =  \frac{1}{\sqrt 2\bar{p}}\left(\frac{ia\sin\theta}{\bar p}-\frac{\cot\theta}{2}+\frac{a^2\sin\theta\cos\theta}{2\rho^2}\right)o_A + \frac{1}{\bar p}\sqrt{\frac{\Delta}{2\rho^2}}\iota_A,\cr
\delta'\iota_A = -\alpha\iota_A + \lambda o_A  = -\frac{1}{\sqrt 2\bar{p}}\left(\frac{ia\sin\theta}{\bar p}-\frac{\cot\theta}{2}+\frac{a^2\sin\theta\cos\theta}{2\rho^2}\right)\iota_A,\cr
\delta o_A = \beta o_A - \tilde{\sigma}\iota_A   = \frac {1}{\sqrt 2 p}\left(\frac{\cot\theta}{2}+\frac{a^2\sin\theta\cos\theta}{2\rho^2}\right)o_A,\cr
\delta\iota_A = -\beta\iota_A + \mu o_A = -\frac {1}{\sqrt 2 p}\left(\frac{\cot\theta}{2}+\frac{a^2\sin\theta\cos\theta}{2\rho^2}\right)\iota_A - \frac{1}{\bar p}\sqrt{\frac{\Delta}{2\rho^2}}\iota_A,\cr
D'o_A = \gamma o_A - \tau\iota_A = \left(\frac{Mr^2 - a^2(r\sin^2\theta + M\cos^2\theta)}{2\rho^2\sqrt{2\Delta \rho^2}} - \frac{ia\cos\theta}{\rho^2}\sqrt{\frac{\Delta}{2\rho^2}}\right)o_A + \frac{ia\sin\theta}{\sqrt 2 \rho^2}\iota_A, \cr
D'\iota_A = -\gamma\iota_A + \nu o_A  = -\left(\frac{Mr^2 - a^2(r\sin^2\theta + M\cos^2\theta)}{2\rho^2\sqrt{2\Delta \rho^2}} - \frac{ia\cos\theta}{\rho^2}\sqrt{\frac{\Delta}{2\rho^2}}\right)\iota_A.
\end{gathered}
\end{equation}
Similarly on the dual conjugation spin-frame $\left\{ o^{A'}, \, \iota^{A'} \right\}$ we have
\begin{equation}\label{SpinDeri2}
\begin{gathered}
Do^{A'} =  \frac{Mr^2 - a^2(r\sin^2\theta + M\cos^2\theta)}{2\rho^2\sqrt{2\Delta\rho^2}} o^{A'} ,\cr
D\iota^{A'} = -\frac{Mr^2 - a^2(r\sin^2\theta + M\cos^2\theta)}{2\rho^2\sqrt{2\Delta\rho^2}}\iota^{A'} + \frac{ia\sin\theta}{\sqrt 2 \bar p^2} o^{A'},\cr
\delta'o^{A'}   = \frac{1}{\sqrt 2\bar{p}}\left(\frac{ia\sin\theta}{\bar p}-\frac{\cot\theta}{2}+\frac{a^2\sin\theta\cos\theta}{2\rho^2}\right) o^{A'} + \frac{1}{\bar p}\sqrt{\frac{\Delta}{2\rho^2}}\iota^{A'},\cr
\delta'\iota^{A'} = -\frac{1}{\sqrt 2\bar{p}}\left(\frac{ia\sin\theta}{\bar p}-\frac{\cot\theta}{2}+\frac{a^2\sin\theta\cos\theta}{2\rho^2}\right) \iota^{A'},\cr
\delta o^{A'} = \frac {1}{\sqrt 2 p}\left(\frac{\cot\theta}{2}+\frac{a^2\sin\theta\cos\theta}{2\rho^2}\right) o^{A'},\cr
\delta\iota^{A'} = -\frac {1}{\sqrt 2 p}\left(\frac{\cot\theta}{2}+\frac{a^2\sin\theta\cos\theta}{2\rho^2}\right)\iota^{A'} + - \frac{1}{\bar p}\sqrt{\frac{\Delta}{2\rho^2}}\iota^{A'},\cr
D'o^{A'} = \left(\frac{Mr^2 - a^2(r\sin^2\theta + M\cos^2\theta)}{2\rho^2\sqrt{2\Delta \rho^2}} - \frac{ia\cos\theta}{\rho^2}\sqrt{\frac{\Delta}{2\rho^2}}\right)o^{A'} + \frac{ia\sin\theta}{\sqrt 2 \rho^2}\iota^{A'}, \cr
D'\iota^{A'} = -\left(\frac{Mr^2 - a^2(r\sin^2\theta + M\cos^2\theta)}{2\rho^2\sqrt{2\Delta \rho^2}} - \frac{ia\cos\theta}{\rho^2}\sqrt{\frac{\Delta}{2\rho^2}}\right)\iota^{A'}.
\end{gathered}
\end{equation}

\subsection{Riemann curvature tensor of Kerr metric}\label{Riemann}
With the usual coordinate transformation $c=\cos\theta$, the Kerr metric \eqref{Metric} becomes
\begin{equation*}
 g=\left(1-\frac{2Mr}{\rho^2}\right)\d t^2 +\frac{4aMr(1-c^2)}{\rho^2}\d t\d \varphi -\frac{\rho^2}{\Delta}\d r^2- \frac{\rho^2}{1-c^2}\d c^2-\frac{\sigma^2}{\rho^2}(1-c^2)\d \varphi^2, 
\end{equation*}
where $\rho^2=r^2+a^2c^2$, $\Delta=r^2-2Mr+a^2$ and $\sigma^2 = (r^2+a^2)\rho^2+ 2Mra^2(1-c^2)$. 

The non zero components of the Riemann curvature tensor are
\begin{eqnarray*}
R_{r c r c} &=&	\frac{3 a^2c^2Mr - Mr^3}{(1 - c^2)\rho^2\Delta},\, R_{rc \varphi t} =	- \frac{ac(a^2c^2M - 3Mr^2)}{\rho^4},\cr
R_{r \varphi r \varphi} &=&	-\frac{(1 - c^2)}{\rho^6 \Delta}(- 9 a^6c^2Mr + 6a^6c^4Mr + 12a^4c^2M^2r^2 - 12a^4c^4M^2r^2)\cr
&& - \frac{(1 - c^2)}{\rho^6 \Delta}(3 a^4Mr^3 - 14a^4c^2Mr^3 + 6a^4c^4Mr^3 - 4a^2M^2r^4)\cr
&& - \frac{(1 - c^2)}{\rho^6 \Delta}( + 4a^2c^2M^2r^4 + 4a^2Mr^5 - 5a^2c^2Mr^5 + Mr^7),\cr
R_{r\varphi r t} &=&\frac{a (1 - c2) (9a^4c^2Mr - 12 a^2c^2M^2r^2 - 3a^2Mr^3 + 9a^2c^2Mr^3 + 4M^2r^4 - 3Mr^5)}
{\rho^6 \Delta},\cr
R_{r\varphi c \varphi} &=& \frac{3a^2c(1 - c^2) (a^2 + r^2) (a^2c^2M - 3Mr^2)}{\rho^6},\cr
R_{r\varphi c t} &=& -\frac{ac(- 3a^2 + 2a^2c^2 - r^2) (a^2c^2M - 3Mr^2)}{\rho^6},\cr
R_{rtrt} &=& -\frac{(-9a^4c^2Mr + 3 a^4c^4Mr + 12a^2c^2M^2r^2 + 3a^2Mr^3 - 7a^2c^2Mr^3 - 4M^2r^4 + 2Mr^5)}{\rho^6\Delta},\cr
R_{rtc\varphi} &=& -\frac{ac(-3a^2 + a^2c^2 - 2r^2) (a^2c^2M - 3Mr^2)}{\rho^6},\, R_{rtct}= \frac{3a^2c(a^2c^2M -3Mr^2)}{\rho^6},\cr
R_{c \varphi c \varphi}	&=& \frac{9a^6c^2Mr + 3a^6c^4Mr + 6a^4c^2M^2r^2 - 6 a^4c^4M^2r^2 + 3a^4Mr^3 - 16a^4c^2Mr^3}{\rho^6} \cr
&& + \frac{3a^4c^4Mr^3 - 2a^2M^2r^4+ 2a^2c^2M^2r^4 + 5a^2Mr^5 - 7a^2c^2Mr^5 + 2Mr^7}{\rho^6},\cr
R_{c \varphi ct} &=& -\frac{a(9a^4c^2Mr - 6a^2c^2M^2r^2 - 3a^2Mr^3 + 9a^2c^2Mr^3 + 2M^2r^4 - 3Mr^5)}{\rho^6},\cr
R_{ctct} &=&	\frac{- 9a^4c^2Mr + 6a^4c^4Mr + 6a^2c^2M^2r^2 + 3a^2Mr^3 - 5a^2c^2Mr^3 - 2M^2r^4 + Mr^5}{(1 - c^2)\rho^6},\cr
R_{\varphi t \varphi t} &=& -\frac{(1 - c^2)\Delta (3a^2c^2Mr - M r^3)}{\rho^6}.
\end{eqnarray*}


\begin{thebibliography}{100}

\bibitem{Ba} D. Batic, {\em Scattering for massive Dirac fields on the Kerr metric}, Journal of Mathematical Physics {\bf 48}, 022502 (2007).

\bibitem{AB} L. Andersson and P. Blue, {\em Hidden symmetries and decay for the wave equation on the Kerr spacetime}, Annals of Mathematics 182 (2015), no. 3, 787-853.

\bibitem{ABJ} L. Andersson, T. B\"ackdahl and J. Joudioux, {\em Hertz potentials and asymptotic properties of massless fields}, Comm. Math. Phys. 331 (2014), no. 2, 755–803.

\bibitem{ABM} L. Andersson, T. B\"ackdahl, P. Blue, and S. Ma, {\em Stability for linearized gravity on the Kerr spacetime}, preprint arXiv:1903.03859 (2019).

\bibitem{Cha} S. Chandrasekhar, {\em The mathematical theory of black holes}, Oxford University Press 1983.

\bibitem{ChruDe2002} P. Chrusciel and E. Delay, {\em Existence of non trivial, asymptotically vacuum, asymptotically simple space-times}, Class. Quantum Grav. {\bf 19} (2002), L71-L79.

\bibitem{CoScho2003} J. Corvino and R.M. Schoen, {\em On the asymptotics for the vacuum Einstein constraint equations}, J. Differential Geom. 73(2): 185--217 (2006).

\bibitem{Da2016} M. Dafermos, I. Rodnianski, and Y. S. -Rothman, {\em Decay for solutions of the wave equation on Kerr exterior spacetimes III: The full subextremal case $|a| < m$}, Annals of Mathematics 183 (2016), no. 3, 787-913.

\bibitem{Da2019} M. Dafermos, G. Holzegel, and I. Rodnianski, {\em Boundedness and decay for the Teukolsky equation on Kerr spacetimes I: the case $|a| \ll m$}, Annals of PDE 5 (2019), no. 1, 2.

\bibitem{FiSm1} F. Finster, N. Kamran, J. Smoller, and S.-T. Yau, {\em Decay rates and probability estimates for massive Dirac particles in the Kerr–Newman black hole geometry}, Communications in Mathematical physics 230 (2002), no. 2, 201-244.

\bibitem{FiSm2} F. Finster, N. Kamran, J. Smoller, and S.-T. Yau, {\em The long-time dynamics of Dirac particles in the Kerr–Newman black hole geometry}, Advances in Theoretical and Mathematical Physics 7 (2003), no. 1, 25-52.

\bibitem{HaNi2004} D. H\"afner and J.-P. Nicolas, {\em Scattering of massless Dirac fields by a Kerr black hole}, Rev. Math. Phys. {\bf 16} (2004), 1, 29--123.

\bibitem{HNM} D. H\"afner, J.-P. Nicolas and M. Mokdad, {\em Scattering theory for Dirac fields inside a Reissner-Nordstr\"om-type black hole}, arXiv:2007.16139 (2021).

\bibitem{Ho1990} L. H\"ormander, {\em A remark on the characteristic Cauchy problem}, J. Funct. Ana. {\bf 93} (1990), 270--277.

\bibitem{Jo2011} J. Joudioux, {\em Integral formula for the characteristic Cauchy problem on a curved background}, J. Math. Pures Appl. (9) 95 (2011), no. 2, 151--193.

\bibitem{Jo2018} J. Joudioux, {\em Gluing for the constraints for higher spin fields}, J. Math. Phys. 58 (11) (2017).

\bibitem{Jo2020} J. Joudioux, {\em H\"ormander’s method for the characteristic Cauchy problem and conformal scattering for a nonlinear wave equation}, Lett. Math. Phys. 110, 1391--1423 (2020).

\bibitem{KaNi2003} S. Klainerman and F. Nicol{\`o}, {\em Peeling properties of asymptotically flat solutions to the Einstein vacuum equations}, Class. Quantum Grav. {\bf 20} (2003), page 3215--3257.

\bibitem{Ma} L.J. Mason, {\em On Ward's integral formula for the wave equation in plane-wave spacetimes}, Twistor Newsletter {\bf 28} (1989), 17--19.

\bibitem{MaNi2004} L.J. Mason and J.-P. Nicolas, {\em Conformal scattering and the Goursat problem}, J. Hyperbolic Differ. Equ. {\bf 1} (2004), 2, 197--233.

\bibitem{MaNi2009} L.J. Mason and J.-P. Nicolas, {\em Regularity at spacelike and null infinity}, J. Inst. Math. Jussieu 8 (2009), 1, 179--208.

\bibitem{MaNi2012} L.J. Mason and J.-P. Nicolas, {\em Peeling of Dirac and Maxwell fields on a Schwarzschild background}, J. Geom. Phys. 62 (2012), no. 4, 867-889.

\bibitem{Ma2020} S. Ma, {\em Almost Price's law in Schwarzschild and decay estimates in Kerr for Maxwell field}, arXiv:2005.12492 (2020).

\bibitem{Ma2021} S. Ma and L. Zhang, {\em Sharp decay estimates for massless Dirac fields on a Schwarzschild background}, 	arXiv:2008.11429 (2021).

\bibitem{MeTaTo2012} J. Metcalfe, D. Tataru, and M. Tohaneanu, {\em Price’s law on nonstationary space–times}, Adv. Math., 230:995-1028, 2012.

\bibitem{MeTaTo2017} J. Metcalfe, D. Tataru, and M. Tohaneanu, {\em Pointwise decay for the Maxwell field on black hole space-times}, Adv. Math., 316:53-93, 2017.

\bibitem{Mo2019} M. Mokdad, {\em Conformal Scattering of Maxwell fields on Reissner-Nordstr\"om-de Sitter Black Hole Spacetimes}, Annales de l'institut Fourier, {\bf 69} (2019), 5, 2291-2329. 

\bibitem{Mo2021} M. Mokdad, {\em Conformal Scattering and the Goursat Problem for Dirac Fields in the Interior of Charged Spherically Symmetric Black Holes}, arXiv:2101.04166 (2021).

\bibitem{Ni95} J.-P.Nicolas, {\em Scattering of linear Dirac fields by a spherically symmetric Black-Hole}, Ann. Inst. Henri Poincar\'e-Physique Th\'eorique, 62 (1995), 2, p. 145-179.

\bibitem{Ni2002} J.-P.Nicolas, {\em Dirac fields on asymptotically flat space-times}, Dissertationes Mathematicae 408, 2002, 85 pages.

\bibitem{Ni2002'} J.-P.Nicolas, {\em A non linear Klein-Gordon equation on Kerr metrics}, Journal de Math\'ematiques Pures et Appliqu\'es,  81 (2002), 9, p. 885--914.

\bibitem{Ni2006} J.-P. Nicolas, {\em On Lars H\"ormander's remark on the characteristic Cauchy problem}, Annales de l'Institut Fourier, {\bf 56} (2006), 3, 517--543.

\bibitem{Ni2015} J.-P. Nicolas, {\em The conformal approach to asymptotic analysis}, Chapter in
''From Riemann to differential geometry and relativity'', Lizhen Ji, Athanase Papadopoulos and Sumio Yamada Eds., Springer, 2017.

\bibitem{Ni2016} J.-P. Nicolas, {\em Conformal scattering on the Schwarzschild metric}, Ann. Inst. Fourier (Grenoble) 66 (2016), no. 3, 1175--1216.

\bibitem{NiXu2019} J.-P. Nicolas and P.T. Xuan, {\em Peeling on Kerr spacetime: linear and nonlinear scalar
fields}, Ann. Henri Poincar\'e 20, 3419–3470 (2019).

\bibitem{Pe63} R. Penrose, {\em Null hypersurface initial data for classical fields of arbitrary spin for general relativity}, Gen. Relativity Gravitation 12 (3) (1980) 225--264 (1963).

\bibitem{Pe1963} R. Penrose, {\em Asymptotic properties of fields and spacetime}, Phys. Rev. Lett. {\bf 10} (1963), 66--68.

\bibitem{Pe1964} R. Penrose, {\em Conformal approach to infinity}, in Relativity, groups and topology, Les Houches 1963, ed. B.S. De Witt and C.M. De Witt, Gordon and Breach, New-York, 1964.

\bibitem{Pe1965} R. Penrose, {\em Zero rest-mass fields including gravitation~: asymptotic behaviour}, Proc. Roy. Soc. London {\bf A284} (1965), 159--203.

\bibitem{PeRi} R. Penrose and W. Rindler, {\em Spinors and space-time}, Vol. I \& II, Cambridge monographs on mathematical physics, Cambridge University Press 1984 \& 1986.

\bibitem{Wa} R. Ward, {\em Progressing waves in flat spacetime and in plane-wave spacetimes}, Class. Quantum Grav. {\bf 4} (1987), 775--778.
    
\bibitem{Whi} E.T. Whittaker, {\em On the partial differential equations of mathematical physics}, Math. Ann. {\bf 57} (1903), p. 333--355.

\bibitem{Sa61} R. Sachs, {\em Gravitational waves in general relativity VI, the outgoing radiation condition}, Proc. Roy. Soc. London {\bf A264} (1961), 309--338.

\bibitem{Ro2020} Y. S.-Rothman and R. T. D. Costa, {\em Boundedness and decay for the Teukolsky equation on Kerr in the full subextremal range $|a|<M$: frequency space analysis}, preprint (2020), arXiv:2007.07211.

\bibitem{Shu} W. T. Shu, {\em Asymptotic Properties of the Solutions of Linear and Nonlinear Spin Field Equations in Minkowski Space}, Commun. Math. Phys. 140, 449--480 (1991).

\bibitem{SmolXi} J. Smoller and Ch. Xie, {\em Asymptotic behavior of massless Dirac waves in Schwarzschild geometry}, Ann. Henri Poincar\'e 13, 943--989 (2012).

\bibitem{StTa} J. Sterbenz and D. Tataru, {\em Local Energy Decay for Maxwell Fields Part I: Spherically Symmetric Black-Hole Backgrounds}, International Mathematics Research Notices, Volume 2015, Issue 11, 2015, Pages 3298–3342.

\bibitem{Sti1936} E. Stiefel, {\em Richtungsfelder und Fernparallelismus in n-dimensionalen Mannigfaltigkeiten}, Comment. Math. Helv. {\bf 8} (1936), 305--353.

\bibitem{Xuan2020} P.T. Xuan, {\em Peeling for Dirac field on Kerr spacetime}, Journal of Mathematical Physics 61, 032501 (2020).

\bibitem{Xuan2021} P.T. Xuan, {\em Conformal scattering theory for the linearized gravity field on Schwarzschild spacetime}, Annals of Global Analysis and Geometry (2021), arXiv:2005.12043.

\bibitem{Xuan20} P.T. Xuan, {\em Conformal scattering theory for a tensorial Fackerell-Ipser equation on Schwarzschild spacetime}, arXiv:2006.02888 (2020).

\bibitem{Xuan21} P.T. Xuan, {\em Cauchy and Goursat problems for the generalized spin-$n/2$ zero rest-mass fields on Minkowski spacetime}, arXiv:2106.04057 (2021).

\end{thebibliography}
\end{document}